\documentclass[journal,twoside,web]{ieeecolor}
\usepackage{generic}
\usepackage{cite}
\usepackage{amsmath,amssymb,amsfonts}
\usepackage{algorithmic}
\usepackage{graphicx}
\usepackage{epstopdf-base}
\usepackage{algorithm,algorithmic}
\usepackage[hidelinks=true]{hyperref}
\usepackage{textcomp}
\usepackage{mathtools}
\usepackage{flushend}
\usepackage{enumerate}

\usepackage{amsthm}
\newtheorem{lemma}{Lemma}[]
\newtheorem{theorem}{Theorem}[]
\newtheorem{condition}{Condition}[]

\newtheorem{proposition}{Proposition}[section]
\theoremstyle{definition}

\newtheorem{oldremark}{Remark}[section]
\newtheorem{oldexample}{Example}[section]

\newtheorem{assumption}{Assumption}[]
\newtheorem{oldcorollary}{Corollary}[]

\newenvironment{corollary}
{\begin{oldcorollary}\pushQED{\qed}}
	{\popQED\end{oldcorollary}}
\newenvironment{remark}
{\begin{oldremark}\pushQED{\qed}}
	{\popQED\end{oldremark}}
\newenvironment{example}
{\begin{oldexample}\pushQED{\qed}}
	{\popQED\end{oldexample}}
\newenvironment{exampleEndEq}
{%
	\begin{oldexample}\pushQED{\qed}%
	}
{\popQED\end{oldexample}}

\usepackage{tikz}
\usetikzlibrary{shapes,arrows}
\usetikzlibrary{arrows,calc,positioning}
\tikzset{
	block/.style = {draw, rectangle,
		minimum height=1.2cm,
		minimum width=1.2cm},
	input/.style = {coordinate,node distance=1cm},
	output/.style = {coordinate,node distance=2cm},
	arrow/.style={draw, -latex,node distance=2cm},
	pinstyle/.style = {pin edge={latex-, black,node distance=1cm}},
	sum/.style = {draw, circle, node distance=1cm}, 
}
\definecolor{backgreen}{HTML}{E9F3DF}
\definecolor{backblue}{HTML}{DAE5EC}
\definecolor{backpurple}{HTML}{E5E0E8}
\definecolor{backorange}{HTML}{F2E9DF}
\definecolor{lightergray}{gray}{0.9}

\usepackage[framemethod=TikZ]{mdframed}
\mdfdefinestyle{callout}{%
	linecolor=black,
	linewidth=1pt,%
	roundcorner=0pt,
	innertopmargin=7pt,
	innerbottommargin=7pt,
	innerrightmargin=7pt,
	innerleftmargin=7pt,
	leftmargin = 0pt,
	rightmargin = 0pt,
	backgroundcolor=lightergray
} 

\def\BibTeX{{\rm B\kern-.05em{\sc i\kern-.025em b}\kern-.08em
    T\kern-.1667em\lower.7ex\hbox{E}\kern-.125emX}}
\usepackage{dsfont}

\usepackage{dirtytalk}

\usepackage{endnotes}

\definecolor{cadetblue}{rgb}{0.33, 0.41, 0.58}

\DeclareMathOperator{\argmin}{\mathrm{argmin}}

\DeclareMathOperator{\EV}{\mathds{E}}

\DeclareMathOperator{\col}{\mathrm{col}}
\DeclareMathOperator{\diag}{\mathrm{diag}}
\DeclareMathOperator{\rank}{\mathrm{rank}}
\DeclareMathOperator{\row}{\mathrm{row}}

\DeclareMathAlphabet{\doublestruck}{U}{BOONDOX-ds}{m}{n}
\newcommand{\ones}[0]{\mathds{1}}
\newcommand{\zeros}[0]{\doublestruck{0}}

\newcommand{\R}[0]{\mathbb{R}}

\newcommand{\Rnn}[0]{\mathbb{R}_{\geq0}}

\newcommand{\Ecal}[0]{\mathcal{E}}

\newcommand{\Pcal}[0]{\mathcal{P}}

\usepackage{xstring}

\renewcommand{\zeros}{\mathbf{0}}
\renewcommand{\ones}{\mathbf{1}}

\begin{document}
\title{Overlapping Covariance Intersection:\\Fusion with Partial Structural Knowledge of Correlation from Multiple Sources}
\author{Leonardo Pedroso, Pedro Batista, W.P.M.H. (Maurice) Heemels \vspace{-0.5cm}
\thanks{The work of L. Pedroso and P. Batista was supported in part by the LARSyS FCT funding (UID/50009/2025: {\small DOI \href{https://doi.org/10.54499/UID/50009/2025}{\texttt{10.54499/UID/50009/2025}}}; LA/P/0083/2020: {\small DOI \href{https://doi.org/10.54499/LA/P/0083/2020}{\texttt{10.54499/LA/P/0083/2020}}}).}
\thanks{L.~Pedroso and W.P.M.H.~Heemels are with the Control Systems Technology section, Department of Mechanical Engineering, Eindhoven University of Technology, The Netherlands (e-mail: \{l.pedroso,w.p.m.h.heemels\}@tue.nl).}
\thanks{L.~Pedroso and P.~Batista are with the Institute for Systems and Robotics, Instituto Superior T\'ecnico, Universidade de Lisboa, Portugal (e-mail: pbatista@isr.tecnico.ulisboa.pt).}}

\maketitle

\begin{abstract}
Emerging large-scale engineering systems rely on distributed fusion for situational awareness, where agents combine noisy local sensor measurements with exchanged information to obtain fused estimates. However, at the sheer scale of these systems, tracking cross-correlations becomes infeasible, preventing the use of optimal filters. Covariance intersection (CI) methods address fusion problems with unknown correlations by minimizing worst-case uncertainty based on available information. Existing CI extensions exploit limited correlation knowledge but cannot incorporate structural knowledge of correlation from multiple sources, which naturally arises in distributed fusion problems. This paper introduces Overlapping Covariance Intersection (OCI), a generalized CI framework that accommodates this novel information structure. We formalize the OCI problem and establish necessary and sufficient conditions for feasibility. We show that a family-optimal solution can be computed efficiently via semidefinite programming, enabling real-time implementation. The proposed tools enable improved fusion performance for large-scale systems while retaining robustness to unknown correlations.

\end{abstract}


\section{Introduction}

Emerging large-scale engineering systems are composed of a large number of agents that interact in an environment to cooperatively achieve a goal. In these settings, each agent has access to noisy data from multiple sensors and from communication with other agents that must be fused to provide a good estimate of the required quantities, for instance, for situational awareness. Examples are mega-constellations of satellites, which rely on absolute position and relative measurements from GNSS receivers and communication to estimate their position  \cite{FergusonHow2003,PedrosoBatista2023DistributedEKF}; vehicle-to-everything networks, where autonomous vehicles obtain data from local sensors and from communication with infrastructure to estimate their position and the position of other vehicles and pedestrians \cite{QiuQiuEtAl2019}. 

Due to the sheer dimension of these systems, it is infeasible to keep track of the correlation between all measurements, which prevents the use of well-known optimal (centralized) filtering solutions \cite{AndersonMoore1979}. As a result, these systems fall into the class of \emph{ultra large-scale systems}, which, by definition, are control systems whose design cannot be feasibly carried out in a centralized manner \cite{PedrosoBatistaEtAl2025ULSS}. Therefore, a distributed fusion approach is required. However, if unknown correlations are ignored, the estimation performance is degraded significantly and, worse, each agent computes deceivingly tighter estimation error bounds than the ground-truth bounds, which can lead to dire consequences. This effect is commonly known as \emph{double-counting} and is showcased in \cite{PanzieriPascucciSetola2006,ChangChongEtAl2010}.

\begin{mdframed}[style=callout]
	\emph{Covariance intersection} (CI)  tools have been devised for the past quarter century to avoid double-counting \cite{ForslingNoackEtAl2024}. These techniques aim at fusing estimates from multiple sources whose correlation is (partially or totally) unknown. The fusion procedure is designed to minimize the worst-case uncertainty of the fused estimate under the available information.
\end{mdframed}

The basic CI setting assumes an information structure whereby the correlation between estimates is totally unknown. It was first developed in \cite{Uhlmann1996,Julier1997}. The review paper \cite{ForslingNoackEtAl2024} provides a recent survey on progress and applications of this method. In the basic setting, when only two estimates are fused, there is a tractable optimal CI method \cite{ReinhardtNoackEtAl2015}, but when dealing with multiple estimates \cite{Uhlmann2003} tractable CI tools do not typically produce an optimal fusion scheme \cite{AjglStraka2018}.

Naturally, partial information about correlations between measurements can significantly improve fusion performance \cite{HanebeckBriechleEtAl2001}. Such partial information can either be inferred from fundamental physical properties of the system or tracked by the agents. Several works in the literature deal with generalizing CI tools to different information structures with partial knowledge about correlation. These are briefly surveyed in Section~\ref{sec:sota}. 

In this paper, we introduce an information structure that has not been analyzed in the CI literature and that stems from the distributed fusion problem over emerging ultra large-scale systems \cite{PedrosoBatistaEtAl2025ULSS}. In Section~\ref{sec:motivating_example}, we introduce the information structure addressed in this paper resorting to a motivating example of a toy cooperative localization problem. Then, in Section~\ref{sec:sota}, the proposed information structure is compared with others previously studied in the literature.

\subsection{Motivating Example}\label{sec:motivating_example}

As a motivating example, consider a cooperative localization toy problem of a team of vehicles. For simplicity, each vehicle $j$ is characterized by a scalar position, denoted by $x_j^k \in \R$ at discrete-time instant $k$. The position of vehicle $j$ evolves according to a discrete-time zero-mean random-walk model with variance $Q_j>0$, which is uncorrelated between different vehicles. Formally, $x_j^{k+1} = x_j^k+d_j^k$, where $d_j^k$ is the drift of the random-walk of vehicle $j$ at time $k$.
The goal is to devise a dynamic distributed filtering solution whereby each vehicle $j$ computes an unbiased estimate of its absolute position at each time instant $k$, which is denoted by $\hat{x}_j^{k} \in \R$. For that, it relies on the unbiased estimate at the previous discrete-time instant, i.e., $\hat{x}_j^{k-1}$, and on relative position measurements w.r.t.\ other neighboring vehicles at time $k$, which are defined in what follows. We denote the estimation error of the position of each generic vehicle $j$ at time $k-1$ by $\tilde{x}^{k-1}_j := \hat{x}^{k-1}_j-x^{k-1}_j$, which is zero-mean because $\hat{x}_j^{k-1}$ is assumed to be unbiased. In this example, we analyze for simplicity a single fusion instance at time $k$ of the vehicle $i$ that is depicted in Fig.~\ref{fig:motivating_example}. Vehicles $p$ and $q$, with respect to which vehicle $i$ gets relative position measurements, are also depicted in Fig.~\ref{fig:motivating_example}.

First, with the information about the dynamical model, vehicle $i$ can make a prediction $z^{k}_{i} := \hat{x}^{k-1}_i$ of its position $x^k_i$ at time $k$. The prediction error $z^{k}_{i}-x^k_i$ is zero-mean because
\begin{equation}\label{eq:eg_pred}
	\begin{split}
		z^{k}_{i} = \hat{x}^{k-1}_i =  x^{k-1}_i +\tilde{x}^{k-1}_i = x_i^{k} + (\tilde{x}^{k-1}_i-d^{k-1}_i)
	\end{split}
\end{equation}
and $\tilde{x}^{k-1}_i$ and $d^{k-1}_i$ are zero-mean.
Second, vehicle $i$, in particular, communicates and has access to relative measurements w.r.t.\ two other neighboring vehicles $p$ and $q$, as depicted in Fig.~\ref{fig:motivating_example}. A relative measurement w.r.t.\ vehicle $p$ at time $k$ is given by $y^{k}_{i,p} := x_p^{k}-x_i^{k} +e^{k}_{i,p}$, where $e^{k}_{i,p}$ is zero-mean measurement noise. Such relative measurement provides an estimate ${z^{k}_{i,p} := \hat{x}^{k-1}_p - y^{k}_{i,p}}$ of $x_i^{k}$, where $\hat{x}_p^{k-1}$ is transmitted from vehicle $p$ to $i$. The error $z^{k}_{i,p} -x_i^{k}$ is zero-mean because
\begin{equation}\label{eq:eg_update}
	\begin{split}
		z^{k}_{i,p} &= \hat{x}^{k-1}_p - y^{k}_{i,p} = x^{k-1}_p +\tilde{x}^{k-1}_p -(x_p^{k} - x_i^{k} + e^{k}_{i,p})  \\
		& = x_p^{k} -d^{k-1}_p+\tilde{x}^{k-1}_p - x_p^{k} + x_i^{k} - e^{k}_{i,p} \\
		& =x_i^{k} + (\tilde{x}_p^{k-1} -d_p^{k-1}-e^{k}_{i,p}),
	\end{split}
\end{equation}
and $\tilde{x}_p^{k-1}$,  $d_p^{k-1}$, and $e^{k}_{i,p}$ are zero-mean. Similarly, a relative measurement $y^{k}_{i,q}$ w.r.t.\ vehicle $q$ provides an estimate $z^{k}_{i,q} := \hat{x}^{k-1}_q - y^{k}_{i,q} = x_i^{k} +(\tilde{x}^{k-1}_q -d^{k-1}_q-e^{k}_{i,q})$.

To obtain a fused estimate of $x_i^{k}$, vehicle $i$ uses a linear filter relying on the information from the predicted estimate and the two relative sensor measurements, i.e., $\hat{x}^{k}_i = K^{k}_{i,i} z^{k}_{i}+ K^{k}_{i,p} z^{k}_{i,p} + K^{k}_{i,q} z^{k}_{i,q}$. This filter is unbiased, i.e., $\EV[\hat{x}^{k}_i] = x_i^{k}$, if and only if $K^{k}_{i,i}+ K^{k}_{i,p} + K^{k}_{i,q}  = 1$. Define $\mathbf{z}^{k}_i := [z^{k}_{i}\;z^{k}_{i,p}\;z^{k}_{i,q}]^\top$, $\mathbf{e}^k_i := [z^{k}_{i}\!-\!x_i^{k}\;\;z^{k}_{i,p}\!-\!x_i^{k}\;\;z^k_{i,q}\!-\!x_i^{k}]^\top$, $\tilde{\boldsymbol{\chi}}_i^{k-1} := [\tilde{x}^{k-1}_p \; \tilde{x}_i^{k-1} \; \tilde{x}_q^{k-1}]^\top$, and $\mathbf{K}^k_i = [K^k_{i,i} \; K^k_{i,p} \; K^k_{i,q} ]$.\footnote{Notice that a non-standard concatenation order is used in the definition of $\tilde{\boldsymbol{\chi}}_i^{k-1}$ for the clarity of the graphical representation of the information structure in Fig.~\ref{fig:structural_info}.} The variance of the estimation error of the fused estimate is given by ${\EV[(\hat{x}_i^{k}-x_i^k)^2]} = {\mathbf{K}_i^k\EV[\mathbf{e}^k_i\mathbf{e}_i^{k\top}]\mathbf{K}_i^{k\top}}$. From \eqref{eq:eg_pred} and \eqref{eq:eg_update}, one can write
\begin{equation}\label{eq:E_ee_motivating_example}
	\begin{split}
	\EV[\mathbf{e}^k_i\mathbf{e}_i^{k\top}] & = \begin{bmatrix}Q_i & \mathbf{0} \\ \mathbf{0}&\mathbf{R}_i^k + \diag(Q_p,Q_q) \end{bmatrix}\\  &+ \begin{bmatrix} 0 & 1 & 0 \\ 1 & 0 & 0\\ 0 & 0 & 1 \end{bmatrix} \EV[\tilde{\boldsymbol{\chi}}_i^{k-1} \tilde{\boldsymbol{\chi}}_i ^{k-1\top}]\begin{bmatrix} 0 & 1 & 0 \\ 1 & 0 & 0\\ 0 & 0 & 1 \end{bmatrix}^\top,
	\end{split}
\end{equation}%
where $\mathbf{R}^k_i = \EV[[e^k_{i,p} \;e^k_{i,q}][e^k_{i,p} \;e^k_{i,q}]^\top]$ is the covariance of the relative measurement noise. The goal for vehicle $i$ is to design the gain $\mathbf{K}_i^k$ such that $\EV[(\hat{x}_i^{k}-x_i^k)^2]$ is minimized.

\begin{mdframed}[style=callout]
	Crucially,  $\EV[\tilde{\boldsymbol{\chi}}_i^{k-1} \tilde{\boldsymbol{\chi}}_i ^{k-1\top}]$ is \emph{not exactly known} to vehicle $i$. Instead, partial information about  $\EV[\tilde{\boldsymbol{\chi}}_i^{k-1} \tilde{\boldsymbol{\chi}}_i ^{k-1\top}]$ is known, and one desires to minimize the worst-case ${\EV[(\hat{x}_i^{k}-x_i^k)^2]}$ under the available information. CI tools have been devised for such filter design problems under different flavors of partial knowledge of correlation between the available estimates.
\end{mdframed}

In the context of this example, we assume that each vehicle $j$ keeps track of the estimate of its own position only and an upper bound on the joint estimation error covariance of the position estimates of vehicle $j$ and its neighbors. In this running example: (i)~vehicle $i$ keeps and updates $\hat{x}_i^{k-1}$ and an upper bound  $\EV[\tilde{\boldsymbol{\chi}}_i^{k-1} \tilde{\boldsymbol{\chi}}_i ^{k-1\top}] \preceq \mathbf{X}^{k-1}_i$; (ii)~vehicle $p$ keeps and updates $\hat{x}_p^{k-1}$ and an upper bound   $\EV[\tilde{\boldsymbol{\chi}}_p^{k-1} \tilde{\boldsymbol{\chi}}_p ^{k-1\top}]\preceq \mathbf{X}_p^{k-1}$, where $\tilde{\boldsymbol{\chi}}^{k-1}_p := [\tilde{x}_i^{k-1} \; \tilde{x}_p^{k-1}]^\top$; and so forth. If the bounds $\mathbf{X}_p^{k-1}$ and $\mathbf{X}_q^{k-1}$ are transmitted from $p$ and $q$, respectively, to $i$, then vehicle $i$ has access to three bounds on principal submatrices of $\EV[\tilde{\boldsymbol{\chi}}_i^{k-1} \tilde{\boldsymbol{\chi}}_i ^{k-1\top}]$. The structure of the bounds $\mathbf{X}_p^{k-1}$ and $\mathbf{X}_q^{k-1}$ on $\EV[\tilde{\boldsymbol{\chi}}_i^{k-1} \tilde{\boldsymbol{\chi}}_i ^{k-1\top}]$ is depicted in Fig.~\ref{fig:structural_info}(v), the bound $\mathbf{X}_i^{k-1}$ would be represented by a bound on the whole $\EV[\tilde{\boldsymbol{\chi}}_i^{k-1} \tilde{\boldsymbol{\chi}}_i ^{k-1\top}]$ and is omitted for the sake of clarity of the figure. The goal for vehicle $i$ is to: (i)~design $\mathbf{K}^k_i$ that minimizes the worst-case second moment of the estimation error of the fused estimate, under knowledge of the structural bounds $\mathbf{X}_i^{k-1}$, $\mathbf{X}_p^{k-1}$, and $\mathbf{X}_q^{k-1}$; and (ii)~compute a consistent bound $\EV[\tilde{\boldsymbol{\chi}}_i^{k} \tilde{\boldsymbol{\chi}}_i ^{k\top}] \preceq \mathbf{X}^{k}_i$ at time $k$.

\begin{figure}[t]
	\centering
	\includegraphics[width = 0.8\linewidth]{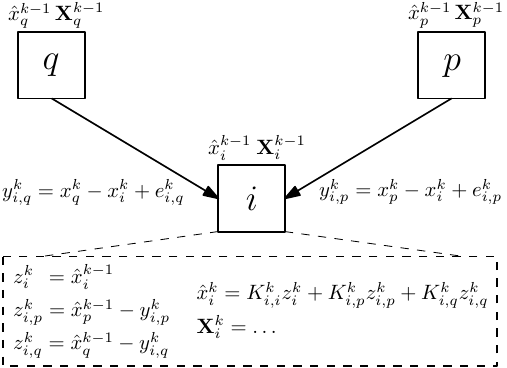}	
	\caption{Scheme of the cooperative localization toy problem.}\label{fig:motivating_example}
\end{figure}

The aforementioned information structure for the fusion problem at vehicle $i$ includes partial structural information about cross-correlations of $\EV[\tilde{\boldsymbol{\chi}}_i^{k-1} \tilde{\boldsymbol{\chi}}_i ^{k-1\top}]$, as depicted in Fig.~\ref{fig:structural_info}(v). In contrast, state-of-the-art approaches to distributed fusion in this toy problem involve receiving only bounds on $\EV[\tilde{x}_p^{k-1} \tilde{x}_p^{k-1\top}]$ and $\EV[\tilde{x}_q^{k-1} \tilde{x}_q^{k-1\top}]$ from vehicles $p$ and $q$, respectively \cite{LiNashashibiEtAl2013}. These bounds only provide information about the autocorrelations of $\EV[\tilde{\boldsymbol{\chi}}_i^{k-1} \tilde{\boldsymbol{\chi}}_i ^{k-1\top}]$, as depicted in Fig.~\ref{fig:structural_info}(ii). Therefore, the novel information structure proposed in this paper has the potential for better distributed fusion performance by making use of partial information about cross-correlations.

\subsection{State-of-the-art}\label{sec:sota}

CI tools have been developed in the literature to address the following information structures, which are schematically represented in Fig.~\ref{fig:structural_info}:
\begin{enumerate}[(i)]
	\item Basic covariance intersection (CI): The correlation between estimates is fully unknown and only \emph{bounds} on the autocorrelation of each estimate are known. After the seminal works \cite{Uhlmann1996} and \cite{Julier1997}, several different approaches have been proposed to address this basic setting (e.g. \cite{ChenArambelEtAl2002}).
	\item Split CI (SCI): The estimates contain an independent component and a (possibly) correlated component. \emph{Bounds} on the independent component are known (matrix on the left in~Fig.~\ref{fig:structural_info}(ii)). Correlation between the correlated components of different estimates is fully unknown, and \emph{bounds} on the autocorrelation of the correlated components of the estimates are known (bounds on matrix on the right in~Fig.~\ref{fig:structural_info}(ii)) \cite{LiNashashibiEtAl2013,WanasingheEtAl2014,JulierUhlmann2017}. SCI was also recently extended to capture known cross-correlations \cite{CrosAmblardEtAl2025}. %
	\item Correlation Coefficient CI (CCCI): A \emph{bound} on the autocorrelation of each estimate is known and a \emph{bound} on a scalar correlation coefficient is also known (e.g., Pearson’s correlation
	coefficient) \cite{HanebeckBriechleEtAl2001,ReeceRoberts2005,WuCaiEtAl2018}.
	\item Partitioned CI (PCI): The autocorrelation of each estimate and the correlation between some pairs of estimates are \emph{exactly} known. The correlation between other pairs is completely unknown \cite{PetersenBeyer2011,AjglStraka2019,AjglStraka2022}.
\end{enumerate}
Motivated by the toy cooperative localization problem in Section~\ref{sec:motivating_example}, we introduce a novel information structure:
\begin{enumerate}[(i)]
	\setcounter{enumi}{4}
	\item Overlapping CI (OCI): Multiple \emph{bounds} on components of the joint estimation error covariance matrix are known (two bounds are represented in Fig.~\ref{fig:structural_info}(v)). The components that are affected by the multiple available bounds may overlap, which justifies the term coined for this structure. In Fig.~\ref{fig:structural_info}(v), the two illustrative bounds are on principal  submatrices, but that need not be the case, as discussed further in Section~\ref{sec:prob_form}. Notice that this is the information structure that we obtain in the toy cooperative localization problem in Section~\ref{sec:motivating_example}. Specifically, bounds $\mathbf{X}_i^{k-1}$, $\mathbf{X}_p^{k-1}$, and $\mathbf{X}_q^{k-1}$ are available on principal submatrices of $\EV[\tilde{\boldsymbol{\chi}}_i^{k-1} \tilde{\boldsymbol{\chi}}_i ^{k-1\top}]$.
\end{enumerate}

\begin{mdframed}[style=callout]
	To the best of the authors' knowledge, the distributed fusion problem under the OCI information structure has not been previously addressed in the literature. It is encountered in the motivating example and it is exemplary of many cooperative localization problems for the next generation of ultra large-scale engineering systems \cite[Section~5]{PedrosoBatistaEtAl2025ULSS}.
\end{mdframed}

\begin{figure}[t]
	\centering
	\includegraphics[width = 0.9 \linewidth]{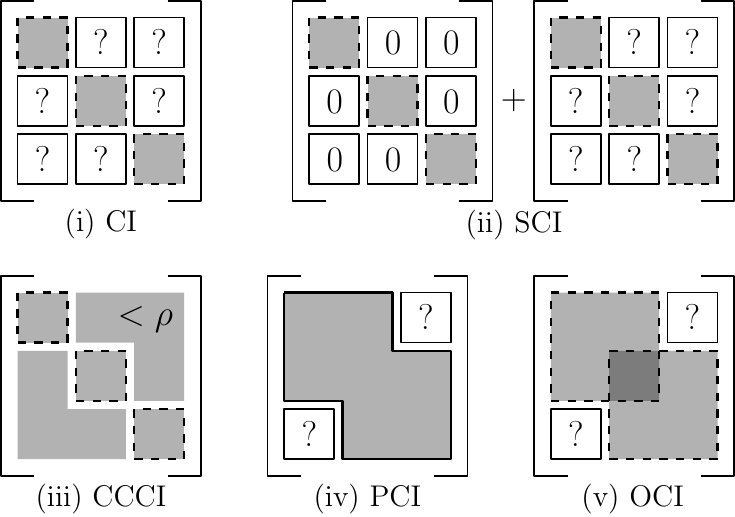}
	\caption{Comparison of the estimation error covariance matrix with three estimates for distinct partial information structures. Dashed contours represent knowledge of bounds and solid contours represent exact knowledge.}\label{fig:structural_info}
\end{figure}

Standard CI approaches are not appropriate to address the OCI problem. Indeed, in Section~\ref{sec:sota_apprch}, we apply a state-of-the-art approach (described, e.g., in \cite{AjglStraka2018}) that is common for CI problems and we analyze its shortcomings.

The PCI problem is the closest in the literature to the OCI problem, since both have \emph{structural knowledge} about the joint estimation error. However, their scope is very different. First, PCI requires \emph{exact} knowledge about element-wise correlation between pairs of estimates, whereas OCI only requires \emph{bounds}. Second, the structural knowledge of the OCI problem is compatible with correlation information that is provided from multiple sources and may overlap as a result, whereas the PCI problem is not. Third, the PCI problem is limited to handling bounds on principal components of the joint estimation error covariance matrix, whereas the OCI problem is not. Moreover, the techniques employed to approach the PCI problem are very different from those used in this paper to address the OCI problem.

\subsection{Contributions}
The main contributions of this paper are twofold:
\begin{enumerate}[(a)]
\item We introduce a distributed fusion problem with \emph{partial structural knowledge of correlation}, which we call OCI, that has not been addressed previously in the literature.
\item The problem is analyzed in depth. We establish necessary and sufficient conditions on the available information to ensure the feasibility of the problem. We express a family-optimal solution to the problem as a semidefinite program (SDP), which is computationally tractable and suitable for real-time implementation.
\end{enumerate}

\subsection{Notation}
Throughout this paper, the $n\times n$ identity, $n\times m$ zero, and $n\times m$ ones matrices are denoted by $\mathbf{I}_n$, $\zeros_{n\times m}$, and $\ones_{n\times m}$, respectively. When clear from context, the subscripts will be dropped to streamline notation.
The sets of $n\times n$ real symmetric positive semidefinite and positive definite matrices are denoted by $S^n_+$ and  $S^n_{++}$, respectively. Moreover, $\mathbf{P} \succ \zeros$ ($\mathbf{P}\succeq \zeros$) denotes that the symmetric matrix $\mathbf{P} \in \R^{n\times n}$ is positive definite (semidefinite) and $\mathbf{P} \succ \mathbf{Q}$ ($\mathbf{P} \succeq \mathbf{Q}$) denotes that the symmetric matrix $\mathbf{P}-\mathbf{Q} \in \R^{n\times n}$  is positive definite (semidefinite). Given a matrix $\mathbf{A}\in \R^{n\times m}$, $\mathbf{A}^+$ denotes the Moore-Penrose pseudoinverse of $\mathbf{A}$ \cite[Chap.~1.6]{Zhang2005}, and $\col{\mathbf{A}} \subseteq \R^{n}$, $\mathrm{row}\,\mathbf{A} \subseteq \R^{m}$, and $\ker{\mathbf{A}} \subseteq \R^{m}$ denote the column space, row space, and kernel of $\mathbf{A}$, respectively. Given a linear subspace $K$ of  $\R^n$, $K^\perp$ denotes the orthogonal complement of $K$ and $\dim{K}$ its dimension.

\section{Problem Formulation}\label{sec:prob_form}

The objective is to estimate a state $\mathbf{x} \in \R^n$ based on the availability of $N$ partial estimates $\mathbf{z}_i = \mathbf{H}_i\mathbf{x}+\mathbf{e}_i$ with $i = 1,2,\ldots,N$ (which come from multiple sources). Here, each $\mathbf{H}_i$ is a known matrix and $\mathbf{e}_i$ is zero-mean random noise. Concatenating all estimates in a single vector $\mathbf{z} := [\mathbf{z}_1^\top,\ldots,\mathbf{z}_N^\top]^\top$, one can write $\mathbf{z} = \mathbf{H}\mathbf{x} + \mathbf{e}$, where $\mathbf{H} := [\mathbf{H}_1^\top \;\; \cdots \;\; \mathbf{H}_N^\top]^\top \in \R^{o\times n}$ and $\mathbf{e} := [\mathbf{e}_1^\top,\ldots,\mathbf{e}_N^\top]^\top$. The second moment of $\mathbf{e}$ is denoted by $\EV[\mathbf{e}\mathbf{e}^\top]$. Naturally, we assume that $\mathbf{z}$ is rich enough to provide an unbiased estimate of $\mathbf{x}$, i.e., that $\mathbf{H}$ has full column rank. We are interested in designing a linear fusion law that provides an unbiased estimator $\hat{\mathbf{x}} = \mathbf{K}\mathbf{z}$, where $\mathbf{K}\in \R^{n\times o}$ is the fusion gain. Furthermore, $\hat{\mathbf{x}}$ is unbiased if $\EV[\hat{\mathbf{x}}] = \mathbf{x}$, which is equivalent to enforcing $\mathbf{K}\mathbf{H} = \mathbf{I}$ when designing $\mathbf{K}$.

Crucially, we consider a case where $ \EV[\mathbf{e}\mathbf{e}^\top]$ is not exactly known and has the form 
\begin{equation}\label{eq:E_ee}
	\EV[\mathbf{e}\mathbf{e}^\top] = \mathbf{R} + \mathbf{C}\mathbf{P}\mathbf{C}^\top,
\end{equation}
where $\mathbf{R}\in S^o_{++}$ and $\mathbf{C} \in \R^{o\times m}$ are known and $\mathbf{P}\in S^m_{++}$ is not exactly known. Specifically, the information structure addressed in this paper assumes knowledge of $M$ bounds on components of $\mathbf{P}$, which come from the sources of the partial estimates. Each bound  $b \in \{1,2, \ldots,M\}$ is written as $\mathbf{W}_b\mathbf{P}\mathbf{W}_b^\top \preceq \mathbf{X}_b$, where $\mathbf{W}_b \in \R^{o_b\times m}$ and $\mathbf{X}_b \in S_{++}^{o_b}$.  One can now define the set of admissible matrices $\mathbf{P}$ given the known bounds as
\begin{equation}\label{eq:Pcal_def}
	\Pcal := \{\mathbf{P} \in S_{++}^m \!:\! \mathbf{W}_b\mathbf{P}\mathbf{W}_b^\top \!\preceq \mathbf{X}_b \;\forall b \in\! \{1,2,\ldots,M\} \}.
\end{equation}

\vspace{0.4cm}
\begin{exampleEndEq}\label{eg:bound_coop}
	Notice that this information structure is a generalization of the one that arises from the cooperative localization toy problem introduced in Section~\ref{sec:motivating_example}. Indeed, there are three estimates available, $z^k_i$, $z^k_{i,p}$, and $z^k_{i,q}$ and $\mathbf{H}_1 = \mathbf{H}_2 = \mathbf{H}_3 = 1$. Moreover, $\EV[\mathbf{e}_i^k\mathbf{e}_i^{k\top}]$ in \eqref{eq:E_ee_motivating_example} has the form of \eqref{eq:E_ee}, with
	\begin{equation*}
		\mathbf{R} = \begin{bmatrix}Q_i & \mathbf{0} \\ \mathbf{0}&\mathbf{R}_i^k + \diag(Q_p,Q_q) \end{bmatrix} \quad \text{and} \quad \mathbf{C} = \begin{bmatrix} 0 & 1 & 0 \\ 1 & 0 & 0\\ 0 & 0 & 1 \end{bmatrix}. 
	\end{equation*}
	Typically, $\mathbf{R}$ represents process and/or sensor noise and $\mathbf{C}$ represents how the uncertainty described by $\mathbf{P}$ shapes the error of the output $\mathbf{z}$. Notice that in the illustrative example $\mathbf{C}$ is merely a permutation matrix, but in generic cooperative localization applications it has a nontrivial definition. The three bounds on  $\mathbf{P} = \EV[\tilde{\boldsymbol{\chi}}_i^{k-1} \tilde{\boldsymbol{\chi}}_i ^{k-1\top}]$ have the form $\mathbf{W}_b\mathbf{P}\mathbf{W}_b^\top \preceq \mathbf{X}_b$, where 
	\begin{equation*}
		\mathbf{W}_1 = \mathbf{I}_3, \; 	\mathbf{W}_2 = \begin{bmatrix} 1 & 0&  0\\ 0 & 1 & 0 \end{bmatrix}, \; \mathbf{W}_3 = \begin{bmatrix} 0 & 1&  0\\ 0 & 0 & 1 \end{bmatrix}.\\\qedhere
	\end{equation*}
\end{exampleEndEq}
\vspace{0.4cm}

The estimation error is a random vector denoted by $\mathbf{\tilde{x}} := \mathbf{\hat{x}}-\mathbf{x}$. Given the constraint $\mathbf{K}\mathbf{H} = \mathbf{I}$ on the gain, one can write $\mathbf{\tilde{x}} =\mathbf{K}\mathbf{z}-\mathbf{x} = \mathbf{K}(\mathbf{H}\mathbf{x}+\mathbf{e})-\mathbf{x} = \mathbf{K}\mathbf{e}$ and $\EV[\mathbf{\tilde{x}}\mathbf{\tilde{x}}^\top] = \mathbf{K}	\EV[\mathbf{e}\mathbf{e}^\top] \mathbf{K}^\top$. The goal is to design a gain $\mathbf{K}$ that optimizes an upper bound on the worst-case second moment of the estimation error for all admissible $\mathbf{P}\in \Pcal$. Formally, the goal is to solve the optimization problem
\begin{equation}\label{eq:OCI_orig_prob}
	\begin{aligned}
		&\min_{\substack{\mathbf{K}\in \R^{n\times o}, \mathbf{B}\in S^n_+}}  && J(\mathbf{B})\\
		&\quad\quad\; \mathrm{s.t.} &&  \mathbf{K}\mathbf{H} = \mathbf{I}\\  
		&&&	\mathbf{B} \succeq \mathbf{K}(\mathbf{R} + \mathbf{C}\mathbf{P}\mathbf{C}^\top)\mathbf{K}^\top,\; \forall \mathbf{P} \in \Pcal,
	\end{aligned}
\end{equation}
where $J: S^n_+ \to \R$ is any optimality criterion that satisfies the following monotonicity condition.

\begin{assumption}\label{ass:J}
	Given $\mathbf{X},\mathbf{Y} \in S^n_+$, the map $J: S^n_+ \to \R$ is such that $\mathbf{X} \succ \mathbf{Y} \implies J(\mathbf{X}) > J(\mathbf{Y})$.
\end{assumption}

This monotonicity assumption on $J$ is very mild. Intuitively, let $\mathbf{B}_1$ and $\mathbf{B}_2$ be error covariance matrices. Assumption~\ref{ass:J} enforces that if $\mathbf{B}_1 \prec \mathbf{B}_2$, i.e., the covariance $\mathbf{B}_2$ portrays a larger spread of the error distribution in every direction than $\mathbf{B}_1$, then $J(\mathbf{B}_2)>J(\mathbf{B}_1)$. Common criteria in fusion applications such as the trace or determinant satisfy it.

\begin{remark}
		A Kalman filtering problem can be cast in this framework. Suppose we have access to an unbiased a priori estimate for $\mathbf{x}$, which is denoted by $\hat{\mathbf{x}}_-$, i.e., $\hat{\mathbf{x}}_- := \mathbf{x} + \mathbf{e}^\prime_-$, where $\mathbf{e}^\prime_-$ is zero-mean noise. In the context of the Kalman filter, $\hat{\mathbf{x}}_-$ would be the so-called predicted estimate. Suppose we also have access to a vector of sensor outputs $\mathbf{y} := \mathbf{C}^\prime\mathbf{x}+\mathbf{e}_y^\prime$, where $\mathbf{e}_y^\prime$ is zero-mean noise. Then, one can formulate the problem of finding an unbiased estimate for $\mathbf{x}$ in the framework presented in this section with $\mathbf{z}^\top = [\hat{\mathbf{x}}_-^\top \; \mathbf{y}^\top]$, $\mathbf{H}^\top = [\mathbf{I}_n \; \mathbf{C}^{\prime\top}]$, and $\mathbf{e}^\top = [ \mathbf{e}^{\prime\top}_- \; \mathbf{e}^{\prime\top}_y]$. If one writes the linear gain $\mathbf{K}$ as $\mathbf{K} = [\mathbf{K}_- \; \mathbf{K}_y]$, where $\mathbf{K}_- \in \R^{n\times n}$, the condition $\mathbf{KH}= \mathbf{I}$ can be equivalently written as $\mathbf{K}_- = \mathbf{I}-\mathbf{K}_y\mathbf{C}^\prime$. The expression for the linear filter then becomes
		\begin{equation*}
			\begin{split}
				\hat{\mathbf{x}} &= \mathbf{K}_-\hat{\mathbf{x}}_- + \mathbf{K}_y\mathbf{y} = \hat{\mathbf{x}}_- + \mathbf{K}_y(\mathbf{y}-\mathbf{C}^\prime \hat{\mathbf{x}}_-),
			\end{split}
		\end{equation*}
		which is the standard form of the update step of a Kalman filter. The design problem becomes finding the gain $\mathbf{K}_y$.
\end{remark}

\begin{remark}
	The form of $\EV[\mathbf{e}\mathbf{e}^\top]$ in \eqref{eq:E_ee} and the partial information structure in \eqref{eq:Pcal_def} generalize multiple CI problems. For the sake of the illustration, consider the availability of only two partial estimates $\mathbf{z}_1$ and $\mathbf{z}_2$. First, a basic CI problem can be cast in the OCI framework with $\mathbf{R} = \zeros$, $\mathbf{C} = \mathbf{I}$, $\mathbf{W}_1 = [\mathbf{I} \; \zeros]$, and $\mathbf{W}_2 = [\zeros \; \mathbf{I}]$. Second, a SCI problem can be cast as $\mathbf{R} = \diag(\mathbf{X}^\mathrm{ind}_1,\mathbf{X}^\mathrm{ind}_2)$, $\mathbf{C} = \mathbf{I}$, $\mathbf{W}_1 = [\mathbf{I} \; \zeros]$, and $\mathbf{W}_2 = [\zeros \; \mathbf{I}]$, where $\mathbf{X}^\mathrm{ind}_1$ and $\mathbf{X}^\mathrm{ind}_2$ are the bounds on the autocorrelation of the independent components. Third, overlapping states fusion with CI \cite{SijsHanebeckEtAl2013} can be cast with $\mathbf{R} = \zeros$ and $\mathbf{C} = \mathbf{I}$ (matrices $\mathbf{W}_1$ and $\mathbf{W}_2$ are omitted for the sake of brevity). However, in this work, we focus exclusively on the case $\mathbf{R} \succ \zeros$, which is the only case of interest for typical fusion problems with a nondegenerate OCI information structure, such as the motivating example in Section~\ref{sec:motivating_example}. The interested reader is referred to \cite{PedrosoBatistaEtAl2026UnificationCI} for a formal analysis on how the OCI problem generalizes other CI problems.
\end{remark}

\section{Overlapping Covariance Intersection}

In this section, we propose an approach to solve the OCI problem \eqref{eq:OCI_orig_prob}. An analysis of the partial knowledge structure is carried out in Section~\ref{sec:analysis}. In Section~\ref{sec:sota_apprch}, we apply a state-of-the-art approach that is common for CI problems to the OCI problem introduced in this paper and discuss why it is not appropriate. In Section~\ref{sec:OCI_reformulation}, we provide a computationally efficient solution approach to the OCI problem \eqref{eq:OCI_orig_prob}.

\subsection{Analysis of Partial Knowledge}\label{sec:analysis}

First, one can rewrite the set $\Pcal$ of admissible matrices $\mathbf{P}$ resorting to bounds on the inverse of $\mathbf{P}$. The convenience of this reformulation will be discussed further in Example~\ref{eg:bounds} and also in Section~\ref{sec:OCI_reformulation}.

\begin{lemma}\label{lem:inverse_bounds}
The set $\Pcal$ of admissible matrices $\mathbf{P}$ can be expressed as $\Pcal := \{\mathbf{P} \in\! S_{++}^m \!:\mathbf{P}^{-1} \!\succeq \mathbf{Y}_b \;\forall b \in \{1,2,\ldots,M\} \}$, where $\mathbf{Y}_b := \mathbf{W}_b^\top\mathbf{X}_b^{-1} \mathbf{W}_b$ for $b = 1,2,\ldots, M$.  
\end{lemma}
\begin{proof}
	See Appendix~\ref{sec:proof_lem_inverse_bounds}.
\end{proof}

Second, we establish necessary and sufficient conditions for the boundedness of admissible covariance matrices. Define $\mathbf{W} := [\mathbf{W}_1^\top \; \cdots \; \mathbf{W}_M^\top]^\top$. The row space of $\mathbf{W}$ characterizes the components where bounds of $\mathbf{P}$ are known, which will be instrumental in the following results.

\begin{lemma}\label{lem:boundedness}
	 There exists $\mathbf{X}\in S_{++}^m$ such that $\mathbf{X} \succeq\mathbf{P}$ for all $\mathbf{P}\in \Pcal$ if and only if $\mathbf{W}$ is full column rank. There exists $\mathbf{Q}\in S_{++}^o$ such that $\mathbf{Q} \succeq \mathbf{R}+\mathbf{C}\mathbf{P}\mathbf{C}^\top$ for all $\mathbf{P}\in \Pcal$ if and only if $\rank(\mathbf{W}) = \rank([\mathbf{W}^\top\; \mathbf{C}^\top]^\top)$.
\end{lemma}
\begin{proof}
	See Appendix~\ref{sec:proof_lem_boundedness}.
\end{proof}

Note that the results of Lemma~\ref{lem:boundedness} are quite intuitive. If the bounds $\mathbf{W}_b\mathbf{P}\mathbf{W}_b^\top \preceq \mathbf{X}_b$, $b = 1,2,\ldots, M$, are informative about all components of  $\mathbf{P}$, i.e., the row space of $[\mathbf{W}_1^\top \; \cdots \; \mathbf{W}_M^\top]^\top$ is $\R^m$ (or, equivalently, $\mathbf{W}$ is full column rank), then the admissible matrices $\mathbf{P}\in\Pcal$ are bounded. Furthermore, if the bounds are informative about the components of $\mathbf{P}$ extracted by $\mathbf{C}$, i.e., the row space of $\mathbf{C}$ is a linear subspace of the row space of $\mathbf{W}$ (or, equivalently,  $\rank(\mathbf{W}) = \rank([\mathbf{W}^\top\; \mathbf{C}^\top]^\top)$), then  $\mathbf{R}+\mathbf{C}\mathbf{P}\mathbf{C}^\top$ is bounded for all admissible $\mathbf{P}\in \Pcal$. Also note that the second condition in Lemma~\ref{lem:boundedness} is weaker than the first, in the sense that the first implies the second.

Third, we turn to the feasibility analysis of the OCI problem \eqref{eq:OCI_orig_prob}. To be clear,  \eqref{eq:OCI_orig_prob} is said to be feasible if there is a pair $(\mathbf{K},\mathbf{B})$ that satisfies the constraints of \eqref{eq:OCI_orig_prob}. It is well-known that $\mathbf{H}$ admits a left inverse if and only if $\mathbf{H}$ is full column rank \cite[Chap.1.3, Lemma~2]{Ben-IsraelGreville2003}, so that is a necessary condition for the existence of $\mathbf{K}$ subject to $\mathbf{K}\mathbf{H} = \mathbf{I}$. On top of that, given any $\mathbf{K}$, the boundedness of $\mathbf{R}+\mathbf{C}\mathbf{P}\mathbf{C}^\top$ for all admissible $\mathbf{P}\in \Pcal$ is a sufficient condition for the existence of $\mathbf{B}$ such that $\mathbf{B} \succeq \mathbf{K}(\mathbf{R} + \mathbf{C}\mathbf{P}\mathbf{C}^\top)\mathbf{K}^\top$ for all $\mathbf{P} \in \Pcal$. As a result, a sufficient condition for feasibility follows as a corollary of Lemma~\ref{lem:boundedness}.

\begin{corollary}\label{cor:feas_suff}
	If $\mathbf{H}$ is full column rank and $\rank(\mathbf{W}) = \rank([\mathbf{W}^\top\; \mathbf{C}^\top]^\top)$, then the OCI problem \eqref{eq:OCI_orig_prob} is feasible.
\end{corollary}

\begin{remark}\label{rm:feas_iff}
	It is possible to establish a necessary and sufficient condition for the feasibility of the OCI problem. Intuitively, this result follows from the fact that the OCI problem is feasible if and only if there is a gain $\mathbf{K}$ that is a left inverse of $\mathbf{H}$ and such that the bounds are informative about the components of $\mathbf{P}$ extracted by $\mathbf{K}\mathbf{C}$, which formally amounts to $\rank(\mathbf{W}) = \rank([\mathbf{W}^\top\; (\mathbf{KC})^\top]^\top)$. However, these conditions cannot be easily written in terms of the parameters $\mathbf{H}$, $\mathbf{R}$, $\mathbf{C}$, and $\mathbf{W}$. Nonetheless, after the reformulation of the OCI problem in Section~\ref{sec:OCI_reformulation}, such a condition follows easily and is stated in Theorem~\ref{th:feas_iff}, which is presented later in the paper.
\end{remark}

\begin{example}\label{eg:bounds}
	The information bounds $\mathbf{P}^{-1} \succeq \mathbf{Y}_b \;\forall b \in \{1,2,\ldots,M\}$ have a very interesting geometrical interpretation. Indeed, a covariance matrix $\mathbf{P}$ can be geometrically represented by an ellipsoid $\Ecal_{\mathbf{P}} = \{\mathbf{x} \in \R^m\!: {\mathbf{x}^\top\mathbf{P}^{-1}\mathbf{x}\leq 1}\}$.
	Furthermore, the relation $\mathbf{P}^{-1} \succeq \mathbf{Y}_b$ can be geometrically understood as $\Ecal_{\mathbf{P}} \subseteq \Ecal_{\mathbf{Y}_b^{-1}}$ \cite{DengZhangEtAl2012}. Therefore, the admissible set $\Pcal$ can be characterized as the intersection of the ellipsoids generated by each bound, i.e., $\mathbf{P}\in \Pcal \iff \Ecal_{\mathbf{P}}\subseteq \cap_{b = 1}^M  \Ecal_{\mathbf{Y}_b^{-1}}$. In Fig.~\ref{fig:eg1_2D} this aspect is illustrated for $m = 2$, with 	$\mathbf{W}_1 = \mathbf{I}_2$, $\mathbf{W}_2 = [1 \; 0]$, $\mathbf{W}_3 =[2 \;-1]$,
	\begin{equation*}
		\mathbf{X}_1 = \begin{bmatrix}
			9 & 6 \\
			6 & 5
		\end{bmatrix}, \quad \mathbf{X}_2 = 1, \quad \text{and} \quad \mathbf{X}_3 = 2.
	\end{equation*}
	Notice that the first bound bounds two components of $\mathbf{P}$, thus it defines a bounded region in $\R^2$. The second and third bounds only bound one component of $\mathbf{P}$, thus they define unbounded regions in $\R^2$ that can be understood as degenerate ellipsoids. The meaning of $\mathbf{Y}^\star$ in Fig.~\ref{fig:eg1_2D} is explained in Example~\ref{eg:sol_OCI}, where this example is revisited.
\end{example}

\begin{figure}[ht]
	\centering
	\includegraphics[width = 0.85\linewidth]{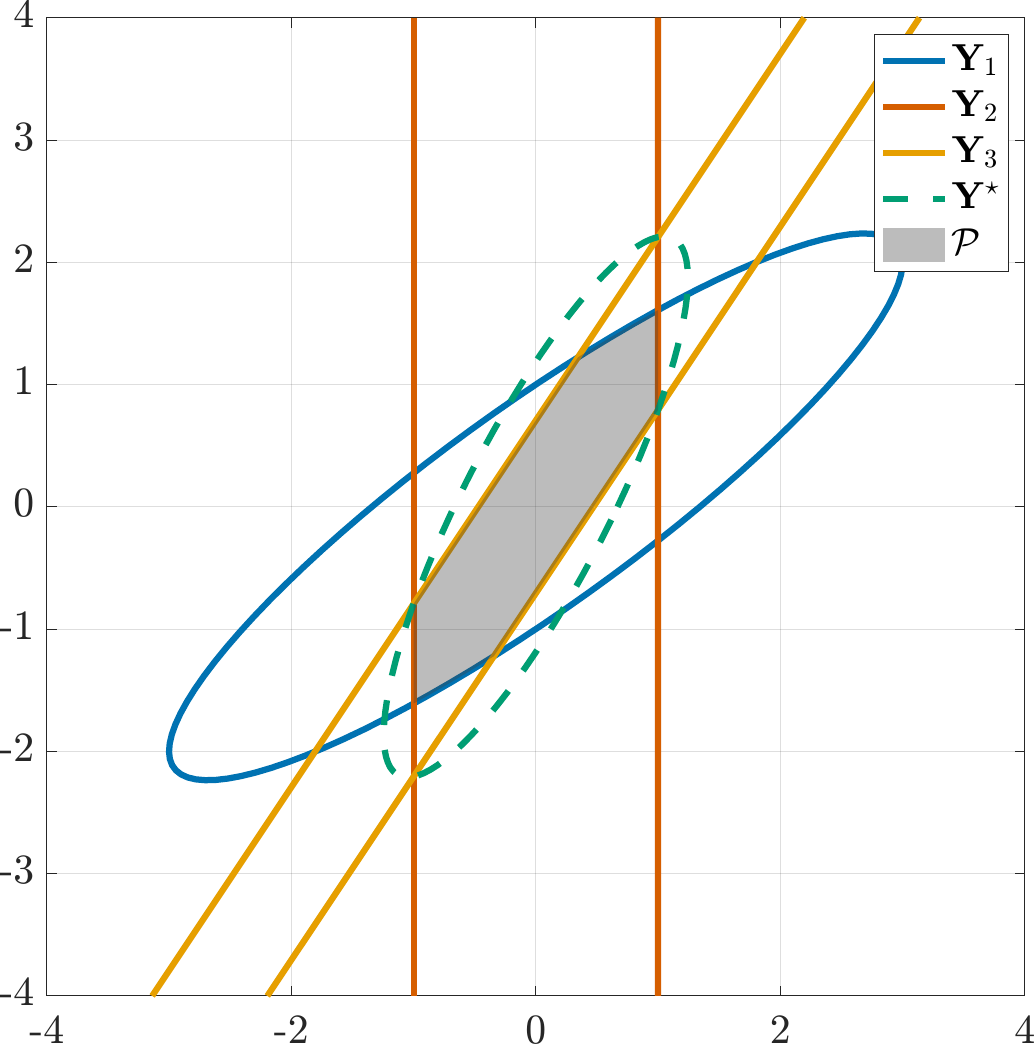}
	\caption{Illustrative OCI bounds, admissible set $\Pcal\subset S_{++}^2$, and trace-minimization OCI solution.}\label{fig:eg1_2D}
\end{figure}

\subsection{A State-of-the-art Approach}\label{sec:sota_apprch}

Henceforth, we resume the analysis of the OCI problem \eqref{eq:OCI_orig_prob} with the goal of finding solutions efficiently. The OCI problem  \eqref{eq:OCI_orig_prob} has a form that is similar to most flavors of the CI problem but suffers from additional complications. In this section, we apply a state-of-the-art approach that is common for CI problems to the OCI problem introduced in this paper and we discuss why it is not appropriate. 

The optimization problem \eqref{eq:OCI_orig_prob} is nonlinear and challenging to solve numerically due to the ``coupling'' between the gain and the admissible matrices $\mathbf{P}\in \Pcal$ in the constraint $\mathbf{B} \succeq \mathbf{K}(\mathbf{R} + \mathbf{C}\mathbf{P}\mathbf{C}^\top)\mathbf{K}^\top,\,  \forall \mathbf{P} \in \Pcal$. Applying the common approach to decouple a CI problem to the OCI problem \eqref{eq:OCI_orig_prob} amounts to proposing a bound $\mathbf{X}$ for $\Pcal$ and then optimizing the gains for that bound with $\mathbf{B} = \mathbf{K}(\mathbf{R}+\mathbf{C}\mathbf{X}\mathbf{C}^\top)\mathbf{K}^\top$, i.e.,
	\begin{equation}\label{eq:OCI_X_bound}
		\begin{aligned}
			&\min_{\substack{\mathbf{K}\in \R^{n\times o}, \mathbf{X}\in S^m_+}}  && J(\mathbf{K}(\mathbf{R} + \mathbf{C}\mathbf{X}\mathbf{C}^\top)\mathbf{K}^\top)\\
			&\quad\quad\;\mathrm{s.t.} &&  \mathbf{K}\mathbf{H} = \mathbf{I}\\  
			&&& \mathbf{X} \succeq \mathbf{P}, \; \forall \mathbf{P} \in \Pcal.
		\end{aligned}
	\end{equation}
	A well-known result under Assumption~\ref{ass:J} is that
	\begin{equation*}
		(\mathbf{H}^\top\!\mathbf{Q}^{-1}\mathbf{H})^{-1}\mathbf{H}^\top\!\mathbf{Q}^{-1} \!= \argmin_\mathbf{K} \;\! J( \mathbf{K}\mathbf{Q}\mathbf{K}^\top\!) \;\; \text{s.t.} \;\;  \mathbf{K}\mathbf{H} \!= \!\mathbf{I}
	\end{equation*}
	for any fixed $\mathbf{Q} \succ \mathbf{0}$. Therefore, \eqref{eq:OCI_X_bound} can be decoupled by optimizing $\mathbf{X}$ when the gain $\mathbf{K}$ is chosen to be the optimal gain, i.e., 
	\begin{equation}\label{eq:OCI_X_bound_noK}
		\begin{aligned}
			&\min_{\mathbf{X}\in S^m_+}  && J((\mathbf{H}^\top(\mathbf{R}+\mathbf{C}\mathbf{X}\mathbf{C}^\top)^{-1}\mathbf{H})^{-1})\\
			&\;\; \mathrm{s.t.} &&  \mathbf{X} \succeq \mathbf{P}, \; \forall \mathbf{P} \in \Pcal.
		\end{aligned}
	\end{equation}
For details and a more thorough analysis of this approach see, for instance, \cite{AjglStraka2018}. There are three main shortcomings of applying this approach to the OCI problem:

\begin{enumerate}[(i)]
	\item Notice that if $\Pcal$ is not bounded, \eqref{eq:OCI_X_bound_noK} is numerically infeasible. One can mitigate this by making use of a change of variables $\mathbf{F} = \mathbf{C}\mathbf{X}\mathbf{C}^\top$, which allows to rewrite \eqref{eq:OCI_X_bound_noK} as $\min_\mathbf{F} \; J((\mathbf{H}^\top\!(\mathbf{R}\!+\!\mathbf{F})^{-1}\mathbf{H})^{-1}) \; \text{s.t.} \; \mathbf{F} \succeq\! \mathbf{C}\mathbf{P}\mathbf{C}^\top\!, \forall \mathbf{P} \!\in\! \Pcal$. Notice that if $\mathbf{C}\mathbf{P}\mathbf{C}^\top\!,\forall \mathbf{P} \!\in\! \Pcal,$ is not bounded, \eqref{eq:OCI_X_bound_noK} is not numerically feasible (even after the change of variables). Still, as analyzed in Remark~\ref{rm:feas_iff}, boundedness of $\mathbf{C}\mathbf{P}\mathbf{C}^\top\!,\forall \mathbf{P} \!\in\! \Pcal,$ is not a necessary condition for the feasibility of the OCI problem. Furthermore, even if $\mathbf{C}\mathbf{P}\mathbf{C}^\top\!,\forall \mathbf{P} \!\in\! \Pcal,$ is bounded, the bound $\mathbf{F}$ that leads to the minimum objective may need to be unbounded in some components, in which case numerical issues will arise. Therefore, expressing the OCI problem by \eqref{eq:OCI_X_bound_noK} reduces the generality of the problem and may lead to numerical issues. In the basic CI setting, there is generally an assumption on the boundedness of the family of admissible covariance matrices.
	
	\item The objective function for the  basic CI problem that is analogous to \eqref{eq:OCI_X_bound_noK} is given by $J((\mathbf{H}^\top\mathbf{X}^{-1}\mathbf{H})^{-1})$. In that case, a variable change $\mathbf{Y} = \mathbf{X}^{-1}$ and thoughtfully devised family of bounds for $\mathbf{Y}$ enables a computationally or analytically tractable (suboptimal) solution \cite{ChenArambelEtAl2002,AjglStraka2018}. Due to the more convoluted objective function of  \eqref{eq:OCI_X_bound_noK}, the same techniques do not work for the OCI problem.
	
	\item Bounds on admissible covariance matrices $\mathbf{P}\in \Pcal$, such as the one in \eqref{eq:OCI_X_bound_noK}, are hard to characterize. In \cite{AjglStraka2017} and \cite{AjglStraka2018}, this issue is analyzed in depth for the basic CI problem (where $\mathbf{P}$ is a block matrix with known diagonal blocks). Indeed, despite addressing a significantly simpler information structure, it is concluded in \cite{AjglStraka2018} that families of bounds of $\mathbf{P}$ are either not \emph{simple} (i.e., not possible to parameterize, thus requiring brute-force design) or do not characterize all \emph{tight} bounds (thus they may not contain the optimal bound, which leads to the suboptimality of the CI optimization problem). For computational efficiency, parametrizations of bounds are often used (see, e.g., \cite{ChenArambelEtAl2002,ReinhardtNoackEtAl2015}). 
\end{enumerate}

For the aforementioned reasons, this approach is not suitable for the OCI problem introduced in this paper, especially because computational efficiency is of  paramount importance in applications such as cooperative localization.

\subsection{Efficient OCI Solution}\label{sec:OCI_reformulation}

In this section, the shortcomings of the state-of-the-art approach outlined in Section~\ref{sec:sota_apprch} are addressed and an efficient procedure to obtain solutions to the OCI problem \eqref{eq:OCI_orig_prob} is proposed. First, we reframe the OCI problem \eqref{eq:OCI_orig_prob} using novel techniques. Second, we introduce a family of bounds of the set of admissible matrices $\Pcal$ in the reframed problem, which allows for a very efficient numerical computation of a \emph{family-optimal} solution resorting to semidefinite programming.

First, given Assumption~\ref{ass:J}, it is possible to decouple the optimization of the gain and covariance bounds in problem~\eqref{eq:OCI_orig_prob}. Indeed, the following result establishes an equivalence between the solutions of \eqref{eq:OCI_orig_prob} and of
\begin{subequations}\label{eq:OCI_Y}
	\begin{alignat}{2}
		&\min_{\!\!\!\!\!\!\mathbf{Y} \in S_{+}^m, \mathbf{U},\mathbf{B}\in S_+^n\!\!\!\!\!}  &\quad & J(\mathbf{B})\\
		&\quad\quad\!\! \mathrm{s.t.} && \begin{bmatrix} \mathbf{B}&  \mathbf{I} \\ \mathbf{I} & \mathbf{H}^\top\mathbf{R}^{-1}\mathbf{H}-\mathbf{U}\end{bmatrix}  \succeq \mathbf{0} \label{eq:OCI_Y_B}\\
		&  && \begin{bmatrix} \mathbf{U}&  \mathbf{H}^\top\mathbf{R}^{-1}\mathbf{C} \\ (\mathbf{H}^\top\mathbf{R}^{-1}\mathbf{C} )^\top & \mathbf{Y}+\mathbf{C}^\top \mathbf{R}^{-1}\mathbf{C}\end{bmatrix}  \succeq \mathbf{0}\label{eq:OCI_Y_U}\\
		&& & \mathbf{Y} \preceq \mathbf{P}^{-1},\; \forall \mathbf{P} \in \Pcal. \label{eq:OCI_Y_Y}
	\end{alignat}
\end{subequations}
To be fully clear, by equivalence we mean that given a solution to \eqref{eq:OCI_Y} one can compute a solution to \eqref{eq:OCI_orig_prob} and vice versa.

\begin{theorem}\label{th:equiv_OCI_prob}
	Assume that Assumption~\ref{ass:J} holds. If $(\mathbf{Y}^\star,\mathbf{U}^\star,\mathbf{B}^\star)$ is a solution to \eqref{eq:OCI_Y}, then the pair $(\mathbf{K}^\star,\mathbf{B}^\star)$ is a solution to \eqref{eq:OCI_orig_prob}, with 
	\begin{equation}\label{eq:K_star_OCI_Y}
		\begin{split}
			\!\!\!\mathbf{K}^\star \! := \! (&\mathbf{H}^\top\!\mathbf{R}^{-1}(\mathbf{R}\!-\!\mathbf{C}(\mathbf{Y}^\star\!+\!\mathbf{C}^\top\! \mathbf{R}^{-1}\mathbf{C})^+ \mathbf{C}^\top\!)\mathbf{R}^{-1}\mathbf{H})^{-1}\!\\
			&\mathbf{H}^\top\!\mathbf{R}^{-1}(\mathbf{R}\!-\!\mathbf{C}(\mathbf{Y}^\star\!+\!\mathbf{C}^\top\! \mathbf{R}^{-1}\mathbf{C})^+ \mathbf{C}^\top\!)\mathbf{R}^{-1}\text{.}\!\!\!
		\end{split}
	\end{equation}
	If $(\mathbf{K}^{\circ},\mathbf{B}^\circ)$ is a solution to \eqref{eq:OCI_orig_prob}, the triple $(\mathbf{Y}^\bullet,\mathbf{U}^\bullet, \mathbf{B}^\bullet)$ is a solution to \eqref{eq:OCI_Y}, with
	\begin{align}
		\mathbf{Y}^\bullet &= (\mathbf{K}^\circ\mathbf{C})^\top(\mathbf{B}^\circ-\mathbf{K}^\circ \mathbf{R}\mathbf{K}^{\circ\top})^+(\mathbf{K}^\circ\mathbf{C})\label{eq:OCI_Y_ball}\\
		\mathbf{U}^\bullet & = \mathbf{H}^\top\mathbf{R}^{-1}\mathbf{C}(\mathbf{Y}^\bullet+\mathbf{C}^\top \mathbf{R}^{-1}\mathbf{C})^+ \mathbf{C}^\top\mathbf{R}^{-1}\mathbf{H}\\
		\mathbf{B}^\bullet  &=	(\mathbf{H}^\top\mathbf{R}^{-1}\mathbf{H}-\mathbf{U}^\bullet)^{-1}.\label{eq:OCI_U_ball}
	\end{align}
	Furthermore, the OCI problem \eqref{eq:OCI_orig_prob} is feasible if and only if  \eqref{eq:OCI_Y} is feasible.
\end{theorem}%
\begin{proof}
	See Appendix~\ref{sec:proof_lem_equiv_OCI_prob}.
\end{proof}

Note that the OCI problem \eqref{eq:OCI_Y} is formulated using inverses of bounds of admissible covariance matrices.\footnote{Within the CI field, the inverse covariance intersection approach studied in \cite{NoackSijsEtAl2017,NygardsDeleskogEtAl2018,AjglStraka2020} also handles inverses of covariance matrices to address the basic CI information structure setting under a specific decomposition of the estimates and covariance bounds. The Kalman filter in information form \cite[Chap.~6]{GrewalAndrews2001} also handles inverses of covariance matrices to prevent numerical issues.} This allows unbounded components to be handled numerically, which addresses shortcoming~(i) described in Section~\ref{sec:sota_apprch}. However, it comes at the expense of a more complex characterization of the bounds $\mathbf{B}$ in the OCI problem \eqref{eq:OCI_orig_prob}. It turns out that resorting to conditions on the positive definiteness of block matrices allows to express the bounds $\mathbf{B}$ resorting to two linear matrix inequalities \eqref{eq:OCI_Y_B} and \eqref{eq:OCI_Y_U}. This enables efficient optimization with off-the-shelf solvers, which addresses shortcoming~(ii).

\begin{remark}
	Besides designing a gain $\mathbf{K}$ and obtaining a fused covariance bound $\mathbf{B}$, one may also be interested in obtaining a bound on $\mathbf{P}$ or on some components of $\mathbf{P}$, i.e., $\mathbf{D}\mathbf{P}\mathbf{D}^\top$ with $\mathbf{D}\in \R^{d\times m}$. In that case, one can add a regularization term to \eqref{eq:OCI_Y} to obtain a good bound with negligible increase of the original objective. Clearly, for such bound $\mathbf{M}\in S_{++}^d$ to exist, $\mathbf{D}\mathbf{P}\mathbf{D}^\top$ needs to be bounded for all $\mathbf{P}\in \Pcal$, which by making cosmetic changes to the proof of Lemma~\ref{lem:boundedness} is the case if and only if $\rank(\mathbf{W}) = \rank([\mathbf{W}^\top\; \mathbf{D}^\top]^\top)$. From Lemma~\ref{lem:inverse_bounds}, it follows that ${\mathbf{D}\mathbf{P}\mathbf{D}^\top \preceq \mathbf{M}} \iff {\mathbf{P}^{-1} \succeq \mathbf{D}^\top\mathbf{M}^{-1}\mathbf{D}}$. So, from $\mathbf{Y}$, one can obtain a bound $\mathbf{M}$ as $\mathbf{Y}\succeq \mathbf{D}^\top\mathbf{M}^{-1}\mathbf{D}$, which by Proposition~\ref{prop:schur}\eqref{schur:Apd} can also be written as a LMI, and one can add a regularization term to the objective to minimize $G(\mathbf{M})$, where $G$ is any performance criterion that also satisfies Assumption~\ref{ass:J}. Then, \eqref{eq:OCI_Y} becomes
	\begin{equation*}
		\begin{aligned}
			&\min_{\substack{\mathbf{Y} \in S_{+}^m, \mathbf{U},\mathbf{B}\in S_+^n\\ \mathbf{M}\in S_{++}^d}}  &\quad & J(\mathbf{B}) + \gamma G(\mathbf{M})\\
			&\quad\quad \;\; \mathrm{s.t.} && \begin{bmatrix} \mathbf{B}&  \mathbf{I} \\ \mathbf{I} & \mathbf{H}^\top\mathbf{R}^{-1}\mathbf{H}-\mathbf{U}\end{bmatrix}  \succeq \mathbf{0} \\
			&  && \begin{bmatrix} \mathbf{U}&  \mathbf{H}^\top\mathbf{R}^{-1}\mathbf{C} \\ (\mathbf{H}^\top\mathbf{R}^{-1}\mathbf{C} )^\top & \mathbf{Y}+\mathbf{C}^\top \mathbf{R}^{-1}\mathbf{C}\end{bmatrix}  \succeq \mathbf{0}\\
			&&&  \begin{bmatrix} \mathbf{M}&  \mathbf{D} \\ \mathbf{D}^\top & \mathbf{Y}\end{bmatrix}  \succeq \mathbf{0}\\
			&& & \mathbf{Y} \preceq \mathbf{P}^{-1},\; \forall \mathbf{P} \in \Pcal,
		\end{aligned}
	\end{equation*}
where $\gamma >0$ is a small regularization weight. For example, in the toy cooperative localization problem presented in Section~\ref{sec:motivating_example}, it is advantageous to obtain a bound $\mathbf{X}_i^{k}$ on $\EV[\tilde{\boldsymbol{\chi}}_i^{k} \tilde{\boldsymbol{\chi}}_i ^{k\top}]$ to be used in the fusion instance at time $k+1$. Such a bound can be obtained using the procedure outlined in this remark.
\end{remark}

Remarkably, the formulation of the OCI problem as \eqref{eq:OCI_Y} yields a necessary and sufficient condition for the feasibility of the original OCI problem \eqref{eq:OCI_orig_prob}, as shown in the following result. Indeed, given $\{\mathbf{W}_b\}_{b \in \{1,2,\ldots,M\}}$, $\mathbf{H}$, $\mathbf{R}$, and $\mathbf{C}$, it is possible to evaluate a simple rank condition to conclude on the feasibility of the OCI problem \eqref{eq:OCI_Y}.
\vspace{-0cm}
\begin{condition}\label{cond:feas}
	The matrix $\mathbf{H}^\top\mathbf{R}^{-1}\mathbf{H}-\mathbf{H}^\top\mathbf{R}^{-1}\mathbf{C}(\mathbf{W}^\top\mathbf{W}+\mathbf{C}^\top \mathbf{R}^{-1}\mathbf{C})^+ \mathbf{C}^\top\mathbf{R}^{-1}\mathbf{H}$ is full rank.
\end{condition}
\vspace{-0.2cm}

\begin{theorem}\label{th:feas_iff}
	Under Assumption~\ref{ass:J}, the OCI problem \eqref{eq:OCI_orig_prob} is feasible if and only if Condition~\ref{cond:feas} holds.
\end{theorem}%
\begin{proof}
	See Appendix~\ref{sec:proof_feas_iff}
\end{proof}

To address shortcoming~(iii) described in Section~\ref{sec:sota_apprch}, similarly to the literature on other flavors of the CI problem, we parameterize a family of bounds for all $\mathbf{P}\! \in\! \Pcal$, i.e., we find a parameterization of a family of matrices $\mathbf{Y}$ that satisfy $\mathbf{Y}\!\preceq\! \mathbf{P}^{-1}$ for all $\mathbf{P}\!\in\! \Pcal$. Recall from Example~\ref{eg:bounds} that each bound $\mathbf{Y}_b$ can be geometrically interpreted as a (possibly degenerate) ellipsoid and that $\Pcal$ can be interpreted as the intersection of all ellipsoids characterized by $\{\mathbf{Y}_b\}_{b \in \{1,2,\ldots,M\}}$. Therefore, finding a family of bounds for all $\mathbf{P} \in \Pcal$ amounts to finding a family of bounds for the intersection of $M$ ellipsoids, also known as a family of circumscribing ellipsoids. We opt to use the very simple family studied in \cite{Kahan1968}, which is characterized by
\begin{equation*}\label{eq:Kahan_parametrization}
	\textstyle\mathbf{Y} = \sum_{b=1}^M\boldsymbol{\omega}_b\mathbf{Y}_b,
\end{equation*}
where $\boldsymbol{\omega} \in \Rnn^{M}$ is a vector that parameterizes the family and must sum to one, i.e.,  $\boldsymbol{\omega} \in  \Delta^M := \{\boldsymbol{\omega} \in \Rnn^{M} : \ones^\top \boldsymbol{\omega} = 1 \}$. We say that a bound $\mathbf{Y}$ is tight if there is no $\mathbf{S}\neq \mathbf{Y}$ with $\mathbf{S}\succeq \mathbf{Y}$ that is also a bound. In \cite{Kahan1968} it is shown that this family characterizes every tight ellipsoid when $M = 2$ and the intersection of the boundary of both ellipsoids is nonempty. For $M>2$ it is also shown that when the intersection of the boundary of all ellipsoids is nonempty (which is rarely the case), each element of the family is a tight bound. However, for $M>2$ the family does not characterize, in general, all tight bounds, therefore an OCI-optimal bound may not be contained in this family. Henceforth, we call this the \emph{Kahan family} of bounding ellipsoids. Remarkably, this family is a generalization of the most common family of bounding ellipsoids used for the basic CI problem, e.g., \cite{ReinhardtNoackEtAl2015}. 

Restricting the OCI problem to this family of bounding ellipsoids amounts to replacing the information constraints $\mathbf{W}_b\mathbf{P}\mathbf{W}_b^\top \preceq \mathbf{X}_b$, with $b = 1,2,\ldots,M$, with a more conservative information constraint $\mathbf{P}^{-1} \succeq  \sum_{b=1}^M\boldsymbol{\omega}_b\mathbf{Y}_b$, where the choice of the parameters $\boldsymbol{\omega}$ only needs to satisfy $\boldsymbol{\omega} \in \Delta^M$. Incorporating the conservative information constraints into the OCI problem \eqref{eq:OCI_orig_prob} leads to the \emph{Kahan-family} OCI problem
\begin{equation}\label{eq:OCI_kahan}
	\begin{aligned}
		&\!\!\!\min_{\substack{\mathbf{K}\in \R^{n\!\times \!o}, \mathbf{B}\in S^n_+\\ \boldsymbol{\omega} \in \Delta^M}} && \!\!\!\! J(\mathbf{B})\\
		&\!\!\!\;\quad\quad \mathrm{s.t.} &&  \!\!\!\!  \mathbf{K}\mathbf{H} = \mathbf{I}\\  
		&&& \!\!\!\! 	\mathbf{B} \succeq \mathbf{K}(\mathbf{R} + \mathbf{C}\mathbf{P}\mathbf{C}^\top)\mathbf{K}^\top\!,\; \forall \mathbf{P} \in \Pcal_{\mathrm{KF}}(\boldsymbol{\omega}) %
	\end{aligned}%
\end{equation}%
where
\begin{equation*}
	\textstyle\Pcal_{\mathrm{KF}}(\boldsymbol{\omega}):= \left\{\mathbf{P} \in S_{++}^m : \mathbf{P}^{-1} \succeq   \sum_{b=1}^M\boldsymbol{\omega}_b\mathbf{Y}_b \right\}.
\end{equation*}
Notice that considering more conservative bounds on the available information via the Kahan family of bounding ellipsoids amounts to tightening the constraint $\mathbf{B}\! \succeq\! \mathbf{K}(\mathbf{R}\! +\! \mathbf{C}\mathbf{P}\mathbf{C}^\top)\mathbf{K}^\top,\, \forall \mathbf{P} \!\in\!\Pcal,$ in the OCI problem \eqref{eq:OCI_orig_prob}. Notice that $M-1$ degrees of freedom have been introduced via $\boldsymbol{\omega}$ in \eqref{eq:OCI_kahan} to parameterize the family. Let $(\mathbf{K}^\star,\mathbf{B}^\star,\boldsymbol{\omega}^\star)$ be a solution to \eqref{eq:OCI_kahan}. We say that $(\mathbf{K}^\star,\mathbf{B}^\star)$ is the \emph{Kahan-family-optimal solution} to the OCI problem \eqref{eq:OCI_orig_prob} associated with the \emph{optimal Kahan bounding ellipsoid} $\mathbf{Y}^\star = \sum_{b=1}^M\boldsymbol{\omega}^\star_b\mathbf{Y}_b$. A computationally efficient characterization of the Kahan-family-optimal solution to the OCI problem \eqref{eq:OCI_orig_prob} follows as a corollary of Theorems~\ref{th:equiv_OCI_prob} and~\ref{th:feas_iff}.

\begin{oldcorollary}\label{cor:OCI_Kahan_optimal}
	The pair $(\mathbf{K}^\star,\mathbf{B}^\star)$ is a Kahan-family-optimal solution to the OCI problem \eqref{eq:OCI_orig_prob}, where $ (\mathbf{U}^\star, \mathbf{B}^\star, \boldsymbol{\omega}^\star)\in$
	\begin{equation}\label{eq:OCI_Kahan_optimal}
		\begin{aligned}
			& \!\!\!\!\underset{\substack{\mathbf{U},\mathbf{B}\in S_+^n\\ \boldsymbol{\omega} \in \Delta^M}}{\argmin} &\quad \!\!\!\!\!&J(\mathbf{B})\\
			&\!\!\!\mathrm{s.t.} && \!\!\begin{bmatrix} \mathbf{B}&  \mathbf{I} \\ \mathbf{I} & \mathbf{H}^\top\mathbf{R}^{-1}\mathbf{H}-\mathbf{U}\end{bmatrix}  \succeq \mathbf{0}\\
			&  && \!\!\begin{bmatrix} \mathbf{U}&  \mathbf{H}^\top\mathbf{R}^{-1}\mathbf{C} \\ (\mathbf{H}^\top\mathbf{R}^{-1}\mathbf{C} )^\top & \!\!\sum_{b=1}^M\!\boldsymbol{\omega}_b\mathbf{Y}_b +\mathbf{C}^\top\! \mathbf{R}^{-1}\mathbf{C}\end{bmatrix} \! \succeq \!\mathbf{0}\!\!\!\!  %
		\end{aligned}
	\end{equation}
	and \vspace{-0.3cm}
	\begin{equation*}
		\begin{split}
				\mathbf{K}^\star \! := \!&\left(\!\!\mathbf{H}^\top\!\mathbf{R}^{-\!1}\!\!\left(\!\!\mathbf{R}\!-\!\mathbf{C}\!\left( \sum_{b=1}^M\!\boldsymbol{\omega}^\star_b\mathbf{Y}_b\!+\!\mathbf{C}^\top\! \mathbf{R}^{-\!1}\mathbf{C}\!\!\right)^{\!\!+}\!\!\! \mathbf{C}^\top\!\!\right)\!\!\mathbf{R}^{\!-\!1}\mathbf{H}\right)^{\!\!\!-\!1}\\
				&\;\;\;\mathbf{H}^\top\!\mathbf{R}^{-\!1}\!\!\left(\!\!\mathbf{R}\!-\!\mathbf{C}\!\left( \sum_{b=1}^M\!\boldsymbol{\omega}^\star_b\mathbf{Y}_b\!+\!\mathbf{C}^\top\! \mathbf{R}^{-\!1}\mathbf{C}\!\!\right)^{\!\!+}\!\!\! \mathbf{C}^\top\!\!\right)\!\!\mathbf{R}^{\!-\!1}.
		\end{split}
	\end{equation*}
Furthermore, \eqref{eq:OCI_Kahan_optimal} is feasible if and only if Condition~\ref{cond:feas} holds or, equivalently, the OCI problem \eqref{eq:OCI_orig_prob} is feasible.
\end{oldcorollary}
\begin{proof}
	See Appendix~\ref{sec:proof_cor_OCI_Kahan_optimal}.
\end{proof}

Crucially, the feasible set of \eqref{eq:OCI_Kahan_optimal} is convex, since it is characterized by two LMIs and a linear equality constraint. Therefore, if $J$ is convex, then \eqref{eq:OCI_Kahan_optimal} is a convex optimization problem and enjoys a plethora of desirable properties such as robustness to changes in input parameters and the existence of efficient numerical algorithms with global optimality guarantees \cite{BoydVandenberghe2004}. Notice that this formulation addresses shortcomings~(ii) and~(iii), pointed out in Section~\ref{sec:sota_apprch}, as a result. Moreover, for common choices of $J$ such as the trace or determinant, one can write \eqref{eq:OCI_Kahan_optimal} as a SDP, for which well-performing off-the-shelf solvers with polynomial worst-case complexity exist \cite{VandenbergheBoyd1996}\cite[Section~6]{MOSEK2024}. 

\begin{remark}
	If one desires to use the determinant as the metric $J$, then the problem \eqref{eq:OCI_Kahan_optimal} needs to be slightly modified to be cast as a SDP. Indeed, since $\det(\mathbf{X}^{-1}) = \det(\mathbf{X})^{-1}$ and the logarithm is strictly increasing, the objective of \eqref{eq:OCI_Kahan_optimal} should be chosen as $-\mathrm{logdet}(\mathbf{H}^\top\mathbf{R}^{-1}\mathbf{H}-\mathbf{U})$ so that it is convex and allows to cast \eqref{eq:OCI_Kahan_optimal} as a SDP (see \cite[Section~6.2.3]{MOSEK2024} for details).
\end{remark}

\begin{example}\label{eg:sol_OCI}
	Solving the Kahan-family-optimal OCI problem \eqref{eq:OCI_Kahan_optimal} with trace minimization for the two-dimensional example in Example~\ref{eg:bounds} with
	\begin{equation*}
		\mathbf{H} = \begin{bmatrix}
			5 \\
			-6
		\end{bmatrix}, \quad \mathbf{C} = \begin{bmatrix}
			-6 & 2 \\
			10 & -4
		\end{bmatrix}, \quad \text{and} \quad \mathbf{R} = \begin{bmatrix}
			8 & 8 \\
			8 & 10
		\end{bmatrix},
	\end{equation*}%
	yields $\mathbf{K}^\star = {[0.1396 \;\; -0.0503]}$ and $\mathbf{B}^\star = 0.9903$. The corresponding optimal Kahan bounding ellipsoid $\mathbf{Y}^\star$ is depicted in Fig.~\ref{fig:eg1_2D}. 
\end{example}

\begin{example}
	One can also compare the fusion performance between OCI and SCI for the toy cooperative localization problem presented in Section~\ref{sec:motivating_example}. Recall that in Example~\ref{eg:bound_coop} the toy cooperative localization problem is expressed as an OCI problem. First, we consider that the information available from vehicle $i$ is characterized by
	\begin{equation*}
		\mathbf{W}_1 = \mathbf{I}_3 \quad \text{and} \quad \mathbf{X}_1 \!= \!\begin{bmatrix}
			25 & \!5 & \!13 \\
			5 & \!\mathbf{9} & \!10 \\
			13 & \!10 & \!25
		\end{bmatrix}.
	\end{equation*}
	We consider that the information available from vehicle $p$ is characterized by 
	\begin{equation*}
		\mathbf{W}_2 = \begin{bmatrix}
			1 & 0 & 0 \\
			0  & 1 & 0
		\end{bmatrix} \quad \text{and} \quad \mathbf{X}_2 \!= \!\begin{bmatrix}
			\mathbf{9} & \!7 \\
	 		7 & \!25
		\end{bmatrix},
	\end{equation*}
	and from vehicle $q$ by
	\begin{equation*}
		\mathbf{W}_3 = \begin{bmatrix} 0 & 1&  0\\ 0 & 0 & 1 \end{bmatrix}\quad \text{and} \quad  \mathbf{X}_3 \!=\! \begin{bmatrix}
			25 & \!-5 \\
			-5 & \!\mathbf{9}
		\end{bmatrix}\!.
	\end{equation*}
	Note that the bound on the variance of the estimation error that each vehicle has of its own position (which is in bold in the equations above) is lower than the one of other vehicles, as expected. We choose $Q_i = Q_p = Q_q = 1$ and $\mathbf{R}_i^k = 4\mathbf{I}_2$. As a result,
	\begin{equation*}
		\mathbf{H} = \begin{bmatrix}
			1 \\
			1 \\
			1
		\end{bmatrix}, \quad \mathbf{R} = \mathrm{diag}(1,5,5), \quad \text{and} \quad \mathbf{C} = \begin{bmatrix} 0 & 1 & 0 \\ 1 & 0 & 0\\ 0 & 0 & 1 \end{bmatrix},
	\end{equation*}
	according to Example~\ref{eg:bound_coop}.
	The OCI solution to this single fusion instance yields $\mathbf{K}^{\star} = [0.5094 \; 0.0610\; 0.4296]$ and $\mathbf{B}^{\star} = 6.5114$.
	The analogous SCI formulation of this problem uses less information, since each vehicle only retains information about a bound on the variance of the estimation error of their own position. Such bound is split into an independent component and a (possibly) correlated component, denoted by $\mathbf{X}_i^\mathrm{ind}$  and $\mathbf{X}_i^\mathrm{dep}$ for vehicle $i$, respectively. As a result, the analogous information that is available for the SCI fusion is characterized by $\mathbf{X}_i^\mathrm{ind} =\mathbf{X}_p^\mathrm{ind} = \mathbf{X}_q^\mathrm{ind} = 9$, $\mathbf{X}_i^\mathrm{dep} = 1$, and $\mathbf{X}_p^\mathrm{dep} = \mathbf{X}_q^\mathrm{dep} = 5$. The SCI solution to this single fusion instance yields $\mathbf{K}^{\star\mathrm{SCI}} = [0.7144 \; 0.1428\; 0.1428]$ and $\mathbf{B}^{\star\mathrm{SCI}} = 9.7143$. One concludes that the set of admissible covariance matrices $\Pcal$ is significantly larger under the information structure of SCI, which explains the significantly larger fused covariance $\mathbf{B}^{\star\mathrm{SCI}}$.
\end{example}

\begin{remark}
	A MATLAB implementation of the computation of the Kahan-family-optimal solution of the OCI problem, as well as the code for all numerical examples in this paper, is available in an open-access repository at \href{https://github.com/decenter2021/OCI}{{\small \texttt{github.com/decenter2021/OCI}}}.
\end{remark}

\section{Conclusion}

The distributed fusion problem addressed in this paper stems from the cooperative localization problem in emerging ultra large-scale engineering systems. In these settings, each agent has access to noisy data from multiple sensors and from communication with other agents that must be fused. It is infeasible to keep track of the covariance between all measurements, but it is feasible that partial structural knowledge about the joint estimation error covariance matrix is tracked by the agents in a distributed fusion framework. The following conclusions are drawn in this paper. First, this problem can be expressed as a generalized covariance intersection (CI) problem and is named overlapping covariance intersection (OCI). Second, we establish necessary and sufficient conditions on the available information for the feasibility of the OCI problem, which turn out to be very mild. Third, we restrict the problem to a parameterized family of bounds for the joint estimation error covariance matrix (which is a generalization of the family of bounds used in the basic CI setting). A solution to the restricted OCI problem is given by a semidefinite program (SDP), which is computationally tractable and suitable for real-time implementation. 

Future work should study whether it is possible to use a less conservative family of bounds while maintaining attractive computational properties. The analysis in \cite{AjglStraka2018} may prove fruitful in that regard. Furthermore, deriving analogous results for the cases $\mathbf{R} = \zeros$ and $\mathbf{R}\succeq \zeros$ would enable the SDP characterization of a family-optimal solution to established CI problems that are a particularization of OCI.

\appendices
\section{Proofs} 

The following results on the positive (semi)definiteness of block matrices will be instrumental in establishing the results in this paper.

\begin{proposition}\label{prop:schur}
	Let $\mathbf{X} = \begin{bmatrix} \mathbf{A} & \mathbf{B} \\  \mathbf{B}^\top & \mathbf{C}
	\end{bmatrix}$ be a symmetric block matrix with $\mathbf{A}\in \R^{p\times p}$ and $\mathbf{C}\in \R^{q\times q}$. Then:
	\begin{enumerate}[(i)]
		\item If $\mathbf{A}\succ \mathbf{0}$, then $\mathbf{X} \succeq \mathbf{0} \iff  \mathbf{C}-\mathbf{B}^\top\mathbf{A}^{-1}\mathbf{B} \succeq \mathbf{0}$;\label{schur:Apd}
		\item   $\mathbf{X} \!\succeq \!\mathbf{0} \iff  \mathbf{A}\succeq \mathbf{0},\; \col{\mathbf{B}}\subseteq \col{\mathbf{A}},\;\mathbf{C}\!-\!\mathbf{B}^\top\mathbf{A}^{+}\mathbf{B} \succeq \mathbf{0}$;\label{schur:Apsd}
		\item If $\mathbf{B} \!=\! \mathbf{I}_p$, then $\mathbf{X}\!\succeq \!\mathbf{0} \! \implies\! {\mathbf{A} \!\succ\! \mathbf{0}},\; {\mathbf{C} \!\succ\! \mathbf{0}},\;\mathbf{A} \!\succeq \! \mathbf{C}^{-1},\; \mathbf{A}^{-1} \preceq \mathbf{C}$;\label{schur:BIr}
		\item  If $\mathbf{B} = \mathbf{I}_p$, then $\mathbf{A} \succ \mathbf{0}, \mathbf{C} \succ \mathbf{0}$ and either $ \mathbf{A} \succeq \mathbf{C}^{-1}$ or $\mathbf{A}^{-1} \preceq \mathbf{C}$ implies $\mathbf{X}\succeq \mathbf{0}$;\label{schur:BIl}
		\item If $\mathbf{X} \succ \mathbf{0}$, then $\mathbf{X}^{-1} \succeq \begin{bmatrix} \mathbf{Q} & \mathbf{0} \\  \mathbf{0} &\mathbf{0}\end{bmatrix} \implies \mathbf{A}^{-1} \succeq \mathbf{Q}$,  where $\mathbf{Q}\in S_{+}^p$.\label{schur:ineq11}
	\end{enumerate}
\end{proposition}
\begin{proof}
	Statement~(i) follows from \cite[Theorem~1.12(b)]{Zhang2005}. Statement~(ii) follows from a particularization of \cite[Theorem~1.20]{Zhang2005}, which holds for any choice of generalized inverse of $\mathbf{A}$, to the Moore-Penrose inverse $\mathbf{A}^+$. Statements~(iii) and~(iv) follow from \cite[Corollary~7.7.10]{HornJohnson2012}. To prove statement~(v), note that, by \cite[Theorem~1.12(a)]{Zhang2005}, since $\mathbf{X}\succ \mathbf{0}$, then $\mathbf{A} \succ \mathbf{0}$, $\mathbf{C} \succ \mathbf{0}$, $\mathbf{C}-\mathbf{B}^\top\mathbf{A}^{-1}\mathbf{B} \succ \mathbf{0}$ and by the Banachiewicz inversion formula \cite[Equation~(0.7.2)]{Zhang2005}, it follows that
	\begin{equation}\label{eq:Banachiewicz}
		\begin{split}
		\mathbf{X}^{-1} &= \begin{bmatrix} \mathbf{A}^{\!-1} & \mathbf{0} \\ \mathbf{0}& \mathbf{0}
		\end{bmatrix}\\&+\begin{bmatrix} -\mathbf{A}^{-1} \mathbf{B}\\ \mathbf{I} \end{bmatrix} (\mathbf{C}-\mathbf{B}^\top\mathbf{A}^{-1}\mathbf{B})^{-1}\begin{bmatrix} -\mathbf{A}^{-1} \mathbf{B}\\ \mathbf{I} \end{bmatrix}^{\top},
		\end{split}
	\end{equation}
	which establishes~(v) immediately.
\end{proof}
\subsection{Proof of Lemma~\ref{lem:inverse_bounds}}\label{sec:proof_lem_inverse_bounds}

To prove the result, one has to show that for any $\mathbf{P}\in S_{++}^m$ and any $b\in \{1,2,\ldots,M\}$, 	$\mathbf{W}_b\mathbf{P}\mathbf{W}_b^\top \preceq \mathbf{X}_b \iff \mathbf{P}^{-1} \succeq \mathbf{Y}_b = \mathbf{W}_b^\top \mathbf{X}_b^{-1}\mathbf{W}_b$. Consider any $\mathbf{P}\in S_{++}^m$ and any $b\in \{1,2,\ldots,M\}$, with $\mathbf{W}_b \in \R^{o_b\times m}$.  Since $\mathbf{P}\succ \mathbf{0}$, by Proposition~\ref{prop:schur}\eqref{schur:Apd}, it follows that $\mathbf{W}_b\mathbf{P}\mathbf{W}_b^\top \preceq \mathbf{X}_b$ is equivalent to 
\begin{equation*}
	\begin{bmatrix} \mathbf{P}^{-1} & \mathbf{W}_b^\top \\  \mathbf{W}_b & \mathbf{X}_b
	\end{bmatrix} \succeq \mathbf{0} \iff  \begin{bmatrix} \mathbf{X}_b  & \mathbf{W}_b \\  \mathbf{W}_b^\top & \mathbf{P}^{-1}
	\end{bmatrix} \succeq \mathbf{0}.
\end{equation*}
Since $\mathbf{X}_b \succ \mathbf{0}$, again by Proposition~\ref{prop:schur}\eqref{schur:Apd}, it follows that $\mathbf{W}_b\mathbf{P}\mathbf{W}_b^\top \preceq \mathbf{X}_b \iff \mathbf{P}^{-1} \succeq \mathbf{W}_b^\top \mathbf{X}_b^{-1}\mathbf{W}_b =  \mathbf{Y}_b$.

\subsection{Proof of Lemma~\ref{lem:boundedness}}\label{sec:proof_lem_boundedness}

The following proposition is required for the proof of the lemma.

\begin{proposition}\label{prop:big_bound_on_P}
	For all $\mathbf{P}\in \Pcal$ it holds that $\mathbf{P}^{-1} \succeq \mathbf{W}^\top\diag(\mathbf{X}_1, \mathbf{X}_2, \ldots,\mathbf{X}_M)^{-1}\mathbf{W}/M$.
\end{proposition}
\begin{proof}
	For any $\mathbf{P} \in \Pcal$, it follows from Lemma~\ref{lem:inverse_bounds} that $\mathbf{P}^{-1} \succeq \mathbf{W}_b^\top \mathbf{X}_b^{-1}\mathbf{W}_b$ for all $b\in \{1,2,\ldots,M\}$. Therefore, $\mathbf{P}^{-1} = \sum_{b = 1}^M \mathbf{P}^{-1}/M \succeq \sum_{b = 1}^{M} \mathbf{W}_b^\top \mathbf{X}_b^{-1}\mathbf{W}_b/M = \mathbf{W}^\top\diag(\mathbf{X}_1, \mathbf{X}_2, \ldots,\mathbf{X}_M)^{-1}\mathbf{W}/M$.
\end{proof}

Notice that the first statement of the lemma is a particular case of the second statement. Indeed, if the second statement holds, then setting $\mathbf{C} = \mathbf{I}$ yields the first statement immediately. Therefore, in what follows, we present a proof only for the second statement of the lemma, which is more general.

One direction of the equivalence states that if there exists $\mathbf{Q}\in S_{++}^o$ such that $\mathbf{Q} \succeq \mathbf{R}+\mathbf{C}\mathbf{P}\mathbf{C}^\top$ for all $\mathbf{P}\in \Pcal$, then $\rank(\mathbf{W}) = \rank([\mathbf{W}^\top\; \mathbf{C}^\top]^\top)$. We prove the reciprocal of this implication, i.e., if $\rank(\mathbf{W}) \neq \rank([\mathbf{W}^\top\; \mathbf{C}^\top]^\top)$, then there is not $\mathbf{Q}\in S_{++}^o$ such that $\mathbf{Q} \succeq \mathbf{R}+\mathbf{C}\mathbf{P}\mathbf{C}^\top$ for all $\mathbf{P}\in \Pcal$. If $\rank(\mathbf{W}) \neq \rank([\mathbf{W}^\top\; \mathbf{C}^\top]^\top)$, then $\rank(\mathbf{W}) < \rank([\mathbf{W}^\top\; \mathbf{C}^\top]^\top)$, which means that there exists a nonnull vector $\mathbf{v} \in \R^m$ such that $\mathbf{W}\mathbf{v} = \mathbf{0}$ and $\mathbf{C}\mathbf{v} \neq \mathbf{0}$. Moreover, from the definition of $\mathbf{W}$, it follows that $\mathbf{W}_b\mathbf{v} = \mathbf{0}$ for all $b\in \{1,2,\ldots,M\}$. Take $\mathbf{P}_1 \in \Pcal$ and define $\mathbf{P}_2 = \mathbf{P}_1 + \alpha \mathbf{v}\mathbf{v}^\top$ for some $\alpha \geq 0$. Notice that for all $b\in \{1,2,\ldots,M\}$, $\mathbf{W}_b\mathbf{P}_2\mathbf{W}_b^\top = \mathbf{W}_b\mathbf{P}_1 \mathbf{W}_b^\top \preceq \mathbf{X}_b$, therefore $\mathbf{P}_2 \in \Pcal$ for any choice of $\alpha\geq 0$. Since $\mathbf{C}\mathbf{v} \neq \mathbf{0}$, there is no $\mathbf{Q}\in S_{++}^o$ such that $\mathbf{R}+\mathbf{C}\mathbf{P}_2\mathbf{C}^\top = \mathbf{R}+\mathbf{C}\mathbf{P}_1\mathbf{C}^\top + \alpha\mathbf{C}\mathbf{v}\mathbf{v}^\top\mathbf{C}^\top\preceq \mathbf{Q}$ holds for any  choice of $\alpha\geq 0$.

We now prove the other direction of the equivalence, i.e., if $\rank(\mathbf{W}) = \rank([\mathbf{W}^\top\; \mathbf{C}^\top]^\top)$, then there exists $\mathbf{Q}\in S_{++}^o$ such that $\mathbf{Q} \succeq \mathbf{R}+\mathbf{C}\mathbf{P}\mathbf{C}^\top$ for all $\mathbf{P}\in \Pcal$.  By Proposition~\ref{prop:big_bound_on_P}, it follows that  $\mathbf{P}^{-1} \succeq \mathbf{W}^\top\diag(\mathbf{X}_1, \mathbf{X}_2, \ldots,\mathbf{X}_M)^{-1}\mathbf{W}/M$ for all $\mathbf{P}\in \Pcal$. Notice that if $\rank(\mathbf{W}) = \rank([\mathbf{W}^\top\; \mathbf{C}^\top]^\top)$, then $\row \mathbf{C} \subseteq \row \mathbf{W}$, hence there exists a matrix $\mathbf{S}$ such that $\mathbf{C} = \mathbf{S}\mathbf{W}$. Since $\diag(\mathbf{X}_1, \mathbf{X}_2, \ldots,\mathbf{X}_M)^{-1} \succ \zeros$, then there exists $\epsilon>0$ such that $\diag(\mathbf{X}_1, \mathbf{X}_2, \ldots,\mathbf{X}_M)^{-1} \succeq \epsilon \mathbf{S}^\top \mathbf{S}$. Therefore, $\mathbf{P}^{-1} \succeq \mathbf{W}^\top\diag(\mathbf{X}_1, \mathbf{X}_2, \ldots,\mathbf{X}_M)^{-1}\mathbf{W}/M \succeq  (\epsilon/M) \mathbf{W}^\top \mathbf{S}^\top \mathbf{S} \mathbf{W} = (\epsilon/M)\mathbf{C}^\top\mathbf{C}$ for all $\mathbf{P}\in \Pcal$. Equivalently, by Lemma~\ref{lem:inverse_bounds}, $\mathbf{C}\mathbf{P}\mathbf{C^\top}\preceq \mathbf{I}M/\epsilon$ for all $\mathbf{P}\in \Pcal$. As a result, defining $\mathbf{Q} = \mathbf{R}+\mathbf{I}M/\epsilon$, it follows that $\mathbf{R} + \mathbf{C}\mathbf{P}\mathbf{C^\top}\preceq \mathbf{Q}$ for all $\mathbf{P}\in \Pcal$, which concludes the proof.

\subsection{Proof of Theorem~\ref{th:equiv_OCI_prob}}\label{sec:proof_lem_equiv_OCI_prob}

The proof of the result relies heavily on the following lemma, which also shows immediately that the OCI problem \eqref{eq:OCI_orig_prob} is feasible if and only if  \eqref{eq:OCI_Y} is feasible.
	
\begin{lemma}\label{lem:feas_lem}
	Assume Assumption~\ref{ass:J} holds. If $(\mathbf{Y}^\star,\mathbf{U}^\star,\mathbf{B}^\star)$ is in the feasible domain of \eqref{eq:OCI_Y}, then the pair $(\mathbf{K}^\star,\mathbf{B}^\star)$  defined in \eqref{eq:K_star_OCI_Y} is in the feasible domain of \eqref{eq:OCI_orig_prob}. If $(\mathbf{K}^{\circ},\mathbf{B}^\circ)$ is in the feasible domain of \eqref{eq:OCI_orig_prob}, the triple $(\mathbf{Y}^\bullet,\mathbf{U}^\bullet, \mathbf{B}^\bullet)$ defined in \eqref{eq:OCI_Y_ball}--\eqref{eq:OCI_U_ball} is in the feasible domain of \eqref{eq:OCI_Y} and  $J(\mathbf{B}^\bullet) \leq J(\mathbf{B}^\circ)$.
\end{lemma}
\begin{proof}
We start by proving the first statement. First, we show that $\mathbf{K}^\star$ is well defined. Applying Proposition~\ref{prop:schur}\eqref{schur:BIr} to \eqref{eq:OCI_Y_B} yields $\mathbf{H}^\top\mathbf{R}^{-1}\mathbf{H}-\mathbf{U}^\star \succ \mathbf{0}$ and applying Proposition~\ref{prop:schur}\eqref{schur:Apsd} to \eqref{eq:OCI_Y_U} gives $\mathbf{U}^\star \succeq  \mathbf{H}^\top\mathbf{R}^{-1}\mathbf{C}(\mathbf{Y}^\star+\mathbf{C}^\top\mathbf{R}^{-1}\mathbf{C})^+\mathbf{C}^\top\mathbf{R}^{-1}\mathbf{H}$. Thus,
\begin{equation}\label{eq:ineq_invK}
	\begin{split}
			\mathbf{H}^\top\mathbf{R}^{-1}(\mathbf{R}-\mathbf{C}(\mathbf{Y}^\star+\mathbf{C}^\top &\mathbf{R}^{-1}\mathbf{C})^+ \mathbf{C}^\top)\mathbf{R}^{-1}\mathbf{H} \\
			&\succeq \mathbf{H}^\top\mathbf{R}^{-1}\mathbf{H}-\mathbf{U}^\star \succ \mathbf{0},
	\end{split}
\end{equation}
showing that $\mathbf{K}^\star$ is well defined. Furthermore, from \eqref{eq:K_star_OCI_Y}, $\mathbf{K}^\star\mathbf{H}  =\mathbf{I}$. Second, we show that $\mathbf{K}^\star(\mathbf{R}+\mathbf{C}\mathbf{P}\mathbf{C}^\top)\mathbf{K}^{\star\top} \preceq \mathbf{B}^\star$ for all $\mathbf{P}\in \Pcal$. By hypothesis, $\mathbf{Y}^\star \preceq \mathbf{P}^{-1}$ for all  $\mathbf{P}\in \Pcal$ and, as a result
\begin{equation}\label{eq:Y_star_cond_aux}
	\mathbf{Y}^\star+\mathbf{C}^\top\mathbf{R}^{-1}\mathbf{C}\preceq \mathbf{P}^{-1}+\mathbf{C}^\top\mathbf{R}^{-1}\mathbf{C}.
\end{equation}
Matrix $\mathbf{Y}^\star+\mathbf{C}^\top\mathbf{R}^{-1}\mathbf{C}$ is real, symmetric, and positive semidefinite, so by the spectral theorem for real symmetric matrices it admits a factorization
\begin{equation}\label{eq:fact_Y_CRC}
\mathbf{Y}^\star+\mathbf{C}^\top\mathbf{R}^{-1}\mathbf{C} = [\mathbf{V}\; \mathbf{V}_\perp]\begin{bmatrix}
	\mathbf{D} & \mathbf{0} \\ \mathbf{0}& {\mathbf{0}}
\end{bmatrix}[\mathbf{V}\; \mathbf{V}_\perp]^\top,
\end{equation} 
where $r= \rank(\mathbf{Y}^\star+\mathbf{C}^\top\mathbf{R}^{-1}\mathbf{C})$, the columns of $[\mathbf{V}\; \mathbf{V}_\perp]$ form an orthonormal basis for $\R^m$, and $\mathbf{D}\in S_{++}^r$ is a diagonal matrix. Furthermore, since $\mathbf{Y}^\star \succeq \mathbf{0}$ and $\mathbf{R}^{-1} \succ \mathbf{0}$, then $\col(\mathbf{C}^\top)\subseteq \col(\mathbf{V})$. To see why this holds, take a vector $\mathbf{v}\in \col(\mathbf{C}^\top)$. Notice that $\mathbf{v}^\top \mathbf{Y}^\star\mathbf{v} \geq 0$ and, since $\mathbf{C} \mathbf{v} \neq \mathbf{0}$, $\mathbf{v}^\top\mathbf{C}^\top \mathbf{R}^{-1} \mathbf{C}\mathbf{v} > 0$. Thus, $\mathbf{v}^\top(  \mathbf{Y}^\star +\mathbf{C}^\top \mathbf{R}^{-1} \mathbf{C})\mathbf{v} = \mathbf{v}^\top \mathbf{V} \mathbf{D} \mathbf{V}^\top \mathbf{v} >0$, which can only hold if $\mathbf{v} \in \col(\mathbf{V})$. One concludes that $\mathbf{v}\in \col(\mathbf{C}^\top) \implies \mathbf{v} \in \col(\mathbf{V})$, i.e., $\col(\mathbf{C}^\top)\subseteq \col(\mathbf{V})$. Moreover, since $\col(\mathbf{C}^\top)\subseteq \col(\mathbf{V})$, $\mathbf{C}\mathbf{V}_\perp = \mathbf{0}$. Using the factorization \eqref{eq:fact_Y_CRC} in \eqref{eq:Y_star_cond_aux} yields
\begin{equation*}
\left([\mathbf{V}\; \mathbf{V}_\perp]^\top	(\mathbf{P}^{-1}+\mathbf{C}^\top\mathbf{R}^{-1}\mathbf{C})^{-1} [\mathbf{V}\; \mathbf{V}_\perp]\right)^{-1} \succeq \begin{bmatrix}
	\mathbf{D} & \mathbf{0} \\ \mathbf{0}& \mathbf{0}
\end{bmatrix}
\end{equation*}
and, by Proposition~\ref{prop:schur}\eqref{schur:ineq11}, it follows that $\mathbf{V}^\top(\mathbf{P}^{-1} +\mathbf{C}^\top\mathbf{R}^{-1}\mathbf{C})^{-1}\mathbf{V} \preceq \mathbf{D}^{-1}$ for all $\mathbf{P}\in \Pcal$, which can equivalently be written as 
\begin{equation}\label{eq:aux_VVp_proofs}
	\begin{split}
	&\;[\mathbf{V} \; \mathbf{V}_\perp] [\mathbf{V} \; \mathbf{0}]^\top (\mathbf{P}^{-1} +\mathbf{C}^\top\mathbf{R}^{-1}\mathbf{C})^{-1} [\mathbf{V} \; \mathbf{0}]	[\mathbf{V} \; \mathbf{V}_\perp]^\top  	\\
	\!\!\preceq &\;[\mathbf{V} \; \mathbf{V}_\perp] \! \begin{bmatrix}
		\mathbf{D}^{-1} &\!\! \mathbf{0} \\ \mathbf{0}& \!\!\mathbf{0}
	\end{bmatrix}  \!	[\mathbf{V} \; \mathbf{V}_\perp]^\top =  (\mathbf{Y}^\star+\mathbf{C}^\top\mathbf{R}^{-1}\mathbf{C})^+.
\end{split}
\end{equation} 
Since $\mathbf{C}\mathbf{V}_\perp = \mathbf{0}$, $\mathbf{C}[\mathbf{V} \; \mathbf{V}_\perp] [\mathbf{V} \; \mathbf{0}]^\top = \mathbf{C}[\mathbf{V} \; \mathbf{V}_\perp] [\mathbf{V} \; \mathbf{V}_\perp]^\top = \mathbf{C}$ then,  pre- and post-multiplying both sides of \eqref{eq:aux_VVp_proofs} by $\mathbf{C}$ and $\mathbf{C}^\top$, respectively, yields $\mathbf{C}(\mathbf{P}^{-1} +\mathbf{C}^\top\mathbf{R}^{-1}\mathbf{C})^{-1}\mathbf{C}^\top \preceq \mathbf{C}(\mathbf{Y}^\star+\mathbf{C}^\top\mathbf{R}^{-1}\mathbf{C})^+ \mathbf{C}^\top$ for all $\mathbf{P}\in \Pcal$.  This condition is equivalent to
\begin{equation}\label{eq:ineq_RU}
	\begin{split}
	(\mathbf{R}+&\mathbf{C}\mathbf{P}\mathbf{C}^\top)^{-1} \\
	&=  \mathbf{R}^{-1}(\mathbf{R}- \mathbf{C}(\mathbf{P}^{-1}+\mathbf{C}^\top\mathbf{R}^{-1}\mathbf{C})^{-1}\mathbf{C}^\top)\mathbf{R}^{-1} \\
	&\succeq \mathbf{R}^{-1}(\mathbf{R}-\mathbf{C}(\mathbf{Y}^\star+\mathbf{C}^\top \mathbf{R}^{-1}\mathbf{C})^+ \mathbf{C}^\top)\mathbf{R}^{-1}
	\end{split}
\end{equation}
for all $\mathbf{P}\in \Pcal$, where the equality follows from the Sherman-Morrison-Woodbury formula \cite[Section~2.1.4]{GolubVanLoan2013} and the inequality follows from \eqref{eq:OCI_Y_Y}. It follows that for all $\mathbf{P} \in \Pcal$
\begin{equation*}
	\begin{split}
		\mathbf{K}^\star&(\mathbf{R}+\mathbf{C}\mathbf{P}\mathbf{C}^\top)\mathbf{K}^{\star\top} \\
		& \preceq (\mathbf{H}^\top\mathbf{R}^{-1}(\mathbf{R}-\mathbf{C}(\mathbf{Y}^\star+\mathbf{C}^\top \mathbf{R}^{-1}\mathbf{C})^+ \mathbf{C}^\top)\mathbf{R}^{-1}\mathbf{H})^{-1}\\
		& \preceq (\mathbf{H}^\top\mathbf{R}^{-1}\mathbf{H}-\mathbf{U}^\star )^{-1}\\
		& \preceq \mathbf{B}^\star,
	\end{split}
\end{equation*}
where the first step follows from algebraic manipulation using \eqref{eq:ineq_RU} and \eqref{eq:K_star_OCI_Y}, the second step follows from \eqref{eq:ineq_invK}, and the third step from Proposition~\ref{prop:schur}\eqref{schur:BIr} applied to \eqref{eq:OCI_Y_B}.
Finally, since $\mathbf{K}^\star\mathbf{H} = \mathbf{I}$, $\mathbf{B}^\star \in S_+^n$, and $\mathbf{K}^\star(\mathbf{R}+\mathbf{C}\mathbf{P}\mathbf{C}^\top)\mathbf{K}^{\star\top} \preceq \mathbf{B}^\star$ for all $\mathbf{P}\in \Pcal$, it follows that $(\mathbf{K}^\star,\mathbf{B}^\star)$ is in the feasible domain of \eqref{eq:OCI_orig_prob}, thereby establishing the first statement of the lemma.

Now, we prove the second statement. First, we show that $\mathbf{P}^{-1} \succeq \mathbf{Y}^\bullet$ for all $\mathbf{P}\in \Pcal$. By hypothesis $\mathbf{K}^\circ(\mathbf{R}+\mathbf{C}\mathbf{P}\mathbf{C}^\top)\mathbf{K}^{\circ\top} \preceq \mathbf{B}^\circ$ for all $\mathbf{P}\in \Pcal$, which can equivalently be written as $\mathbf{K}^\circ\mathbf{C}\mathbf{P}(\mathbf{K}^\circ\mathbf{C})^\top \preceq \mathbf{B}^\circ -\mathbf{K}^\circ\mathbf{R}\mathbf{K}^{\circ\top}:= \tilde{\mathbf{B}}^\circ$. By Proposition~\ref{prop:schur}\eqref{schur:Apd}, this is equivalent to
\begin{equation*}
	\begin{bmatrix} \mathbf{P}^{-1} & (\mathbf{K}^\circ\mathbf{C})^\top \\  \mathbf{K}^\circ\mathbf{C}&  \tilde{\mathbf{B}}^\circ
	\end{bmatrix} \succeq \mathbf{0} \iff 	 \begin{bmatrix} \tilde{\mathbf{B}}^\circ  & \mathbf{K}^\circ\mathbf{C}\\  (\mathbf{K}^\circ\mathbf{C})^\top& \mathbf{P}^{-1}
	\end{bmatrix} \succeq \mathbf{0} .
\end{equation*}
Since $\tilde{\mathbf{B}}^\circ \succeq \mathbf{0}$, by Proposition~\ref{prop:schur}\eqref{schur:Apsd}, it follows that  $\mathbf{P}^{-1}\succeq (\mathbf{K}^\circ\mathbf{C})^\top\tilde{\mathbf{B}}^{\circ+} \mathbf{K}^\circ\mathbf{C} = \mathbf{Y}^\bullet$ for all $\mathbf{P}\in \Pcal$ and $\col{\mathbf{K}^\circ\mathbf{C}}\subseteq \col{\tilde{\mathbf{B}}^\circ}$.
Second, we show that $\mathbf{B}^\circ \succeq \mathbf{K}^\circ(\mathbf{R}+\mathbf{C} \mathbf{Y}^{\bullet+}\mathbf{C}^\top)\mathbf{K}^{\circ\top}$. To do so, notice that since $\tilde{\mathbf{B}}^\circ \succeq \mathbf{0}$, $\col{\mathbf{K}^\circ\mathbf{C}}\subseteq \col{\tilde{\mathbf{B}}^\circ}$, and $ \mathbf{Y}^\bullet-(\mathbf{K}^\circ\mathbf{C})^\top\tilde{\mathbf{B}}^{\circ+} \mathbf{K}^\circ\mathbf{C} = \mathbf{0} \succeq \mathbf{0}$, then by Proposition~\ref{prop:schur}\eqref{schur:Apsd} 
\begin{equation*}
	\begin{bmatrix} \tilde{\mathbf{B}}^\circ  & \mathbf{K}^\circ\mathbf{C}\\  (\mathbf{K}^\circ\mathbf{C})^\top& \mathbf{Y}^\bullet
	\end{bmatrix} \succeq \mathbf{0} \iff \begin{bmatrix}  \mathbf{Y}^\bullet& (\mathbf{K}^\circ\mathbf{C})^\top \\  \mathbf{K}^\circ\mathbf{C}&  \tilde{\mathbf{B}}^\circ
	\end{bmatrix} \succeq \mathbf{0}.
\end{equation*}
Using Proposition~\ref{prop:schur}\eqref{schur:Apsd}, it follows that $ \tilde{\mathbf{B}}^\circ \succeq \mathbf{K}^\circ\mathbf{C} \mathbf{Y}^{\bullet+}(\mathbf{K}^\circ\mathbf{C})^\top$, which is equivalent to $\mathbf{B}^\circ \succeq \mathbf{K}^\circ(\mathbf{R}+\mathbf{C} \mathbf{Y}^{\bullet+}\mathbf{C}^\top)\mathbf{K}^{\circ\top}$. Third, we show that $\mathbf{K}(\mathbf{R}+\mathbf{C}\mathbf{Y}^{\bullet+}\mathbf{C}^\top)\mathbf{K}^\top \succeq  \mathbf{K}(\mathbf{R}+\mathbf{C}\mathbf{P}\mathbf{C}^\top)\mathbf{K}^\top$ for all $\mathbf{P} \in \Pcal$ and all $\mathbf{K}$ that satisfy a condition on their kernel. Matrix $\mathbf{Y}^\bullet$ is real, symmetric, and positive semidefinite so, by the spectral theorem for real symmetric matrices, it admits a factorization
\begin{equation}\label{eq:diag_Y_circ}
	\mathbf{Y}^\bullet= [\mathbf{V}\; \mathbf{V}_\perp]\begin{bmatrix}
		\mathbf{D} & \mathbf{0} \\ \mathbf{0}_{(m-r)\times r} & \mathbf{0}
	\end{bmatrix}[\mathbf{V}\; \mathbf{V}_\perp]^\top,
\end{equation} 
where $r= \rank(\mathbf{Y}^\bullet)$, the columns of $[\mathbf{V}\; \mathbf{V}_\perp]$ form an orthonormal basis for $\R^m$, and $\mathbf{D}\in S_{++}^r$ is a diagonal matrix. 
Furthermore, since $\tilde{\mathbf{B}}^{\circ}  \succeq \mathbf{K}^\circ \mathbf{C}\mathbf{P}( \mathbf{K}^\circ \mathbf{C})^\top$ by hypothesis, then $\tilde{\mathbf{B}}^{\circ+}$ can only possibly have null eigenvalues along the components in $\ker{(\mathbf{K}^\circ \mathbf{C})^\top} = (\col \mathbf{K}^\circ \mathbf{C})^\perp$ (otherwise $\mathbf{x}^\top(\tilde{\mathbf{B}}^{\circ}  - \mathbf{K}^\circ \mathbf{C}\mathbf{P}( \mathbf{K}^\circ \mathbf{C})^\top)\mathbf{x} = -\mathbf{x}^\top \mathbf{K}^\circ \mathbf{C}\mathbf{P}( \mathbf{K}^\circ \mathbf{C})^\top\mathbf{x} < 0$ for $\mathbf{x} \in \ker \tilde{\mathbf{B}}^{\circ}  \setminus \ker{(\mathbf{K}^\circ \mathbf{C})^\top}$, which is not possible). Moreover, since $\mathbf{Y}^\bullet = (\mathbf{K}^\circ\mathbf{C})^\top\tilde{\mathbf{B}}^{\circ+} \mathbf{K}^\circ\mathbf{C}$ and the null components of $\tilde{\mathbf{B}}^{\circ+}$  can only possibly be in $(\col \mathbf{K}^\circ \mathbf{C})^\perp$, the null components of $\mathbf{Y}^\bullet$ are along $\ker \mathbf{K}^\circ\mathbf{C}$. As a result, $\mathbf{K}^\circ\mathbf{C}\mathbf{V}_\perp = \mathbf{0}$. Employing the same procedure as in step two of the proof of the first statement of the lemma for the expression $\mathbf{P}^{-1} \succeq \mathbf{Y}^\bullet$ yields 
\begin{equation}\label{eq:aux1}
	[\mathbf{V} \; \mathbf{V}_\perp] [\mathbf{V} \; \mathbf{0}]^\top \mathbf{P}[\mathbf{V} \; \mathbf{0}]	[\mathbf{V} \; \mathbf{V}_\perp]^\top \preceq \mathbf{Y}^{\bullet+}
\end{equation} %
for all $\mathbf{K}$ such that $\mathbf{K}\mathbf{C}\mathbf{V}_\perp = \mathbf{0}$, $\mathbf{K}\mathbf{C}[\mathbf{V} \; \mathbf{V}_\perp] [\mathbf{V} \; \mathbf{0}]^\top = \mathbf{K}\mathbf{C}[\mathbf{V} \; \mathbf{V}_\perp] [\mathbf{V} \; \mathbf{V}_\perp]^\top = \mathbf{K}\mathbf{C}$. Therefore, pre- and post-multiplying both sides of \eqref{eq:aux1} by $\mathbf{KC}$ and $(\mathbf{KC})^\top$, respectively, yields $\mathbf{KC} \mathbf{P} (\mathbf{KC})^\top \preceq \mathbf{KC}\mathbf{Y}^{\bullet+}(\mathbf{KC})^\top$.  This condition is equivalent to
\begin{equation}\label{eq:KCPCP}
	\mathbf{K}(\mathbf{R}+\mathbf{C}\mathbf{P}\mathbf{C}^\top) \mathbf{K}^\top \preceq 	\mathbf{K}(\mathbf{R}+\mathbf{C}\mathbf{Y}^{\bullet+}\mathbf{C}^\top) \mathbf{K}^\top
\end{equation}
for all $\mathbf{K}$ such that $\mathbf{K}\mathbf{C}\mathbf{V}_\perp = \mathbf{0}$ and all $\mathbf{P}\in \Pcal$.
Fourth, we show that $\mathbf{B}^\bullet \succeq (\mathbf{H}^\top\mathbf{R}^{-1}(\mathbf{R}-\mathbf{C}(\mathbf{Y}^\bullet+\mathbf{C}^\top \mathbf{R}^{-1}\mathbf{C})^+ \mathbf{C}^\top)\mathbf{R}^{-1}\mathbf{H})^{-1}$.  
Consider the optimization problem 
\begin{equation}\label{eq:OCI_mod_prob}
	\begin{aligned}
		&\min_{\substack{\mathbf{K}\in \R^{n\times o}, \mathbf{B}\in S^n_+}}  && J(\mathbf{B})\\
		&\quad\quad\;\mathrm{s.t.} &&  \mathbf{K}\mathbf{H} = \mathbf{I}\\  
		&&& \mathbf{K}\mathbf{C}\mathbf{V}_\perp = \mathbf{0}\\
		&&&	\mathbf{B} \succeq \mathbf{K}(\mathbf{R} + \mathbf{C}\mathbf{Y}^{\bullet+}\mathbf{C}^\top)\mathbf{K}^\top.
	\end{aligned}
\end{equation}
Notice that the feasible set of \eqref{eq:OCI_mod_prob} is contained in the feasible set of the original OCI optimization problem \eqref{eq:OCI_orig_prob}. Indeed, if the pair $(\mathbf{K},\mathbf{B})$ satisfies the constraints of \eqref{eq:OCI_mod_prob}, then $\mathbf{KH} = \mathbf{I}$. Since $\mathbf{K}\mathbf{C}\mathbf{V}_\perp = \mathbf{0}$, by \eqref{eq:KCPCP}, it follows that $\mathbf{K}(\mathbf{R}+\mathbf{C}\mathbf{P}\mathbf{C}^\top) \mathbf{K}^\top \preceq 	\mathbf{K}(\mathbf{R}+\mathbf{C}\mathbf{Y}^{\bullet+}\mathbf{C}^\top) \mathbf{K}^\top \preceq \mathbf{B}$. Define $\mathbf{S}_\perp$ as a matrix whose columns form an orthonormal basis for $\mathbf{C}\mathbf{V}_\perp$ and $\mathbf{S}$ such that the columns of $[\mathbf{S}\; \mathbf{S}_\perp]$ form an orthonormal basis for $\R^o$. A gain that satisfies $ \mathbf{K}\mathbf{C}\mathbf{V}_\perp = \mathbf{0}$ can be written as $\mathbf{K} = \mathbf{K}[\mathbf{S}\; \mathbf{S}_\perp][\mathbf{S}\; \mathbf{S}_\perp]^\top = [\mathbf{K}\mathbf{S}\; \mathbf{0}][\mathbf{S}\; \mathbf{S}_\perp]^\top = \mathbf{K}\mathbf{S}\mathbf{S}^\top$. One can then equivalently rewrite \eqref{eq:OCI_mod_prob} for $\tilde{\mathbf{K}} =  \mathbf{K}\mathbf{S}$ as
\begin{equation}\label{eq:OCI_mod_prob2}
	\begin{aligned}
		&\min_{\substack{\tilde{\mathbf{K}}, \mathbf{B}\in S^n_+}}  && J(\mathbf{B})\\
		&\quad\mathrm{s.t.} &&  \tilde{\mathbf{K}}\mathbf{S}^\top\mathbf{H} = \mathbf{I}\\  
		&&&	\mathbf{B} \succeq \tilde{\mathbf{K}}\mathbf{S}^\top(\mathbf{R} + \mathbf{C}\mathbf{Y}^{\bullet+}\mathbf{C}^\top)\mathbf{S}\tilde{\mathbf{K}}^\top
	\end{aligned}
\end{equation}
and then recover $\mathbf{K}$ with the relation $\mathbf{K} = \tilde{\mathbf{K}}\mathbf{S}^\top$.
Given Assumption~\ref{ass:J}, it is well-known \cite{AjglStraka2018} that, for any $\mathbf{Q}\in S_{++}^o$,
	\begin{equation*}
	(\mathbf{H}^\top\!\mathbf{Q}^{-1}\mathbf{H})^{-1}\mathbf{H}^\top\!\mathbf{Q}^{-1} \!= \argmin_\mathbf{K} \;\! J( \mathbf{K}\mathbf{Q}\mathbf{K}^\top\!) \;\; \text{s.t.} \;\;  \mathbf{K}\mathbf{H} \!= \!\mathbf{I}.
\end{equation*}
As a result, the pair $(\mathbf{K}^\bullet, \mathbf{B}^\bullet)$ with
\begin{align*}
	\mathbf{K}^\bullet  &= \left(\mathbf{H}^\top\mathbf{S}\left(\mathbf{S}^\top(\mathbf{R} + \mathbf{C}\mathbf{Y}^{\bullet+}\mathbf{C}^\top)\mathbf{S}\right)^{-1}\mathbf{S}^\top\mathbf{H}\right)^{-1}\\
	&\quad\quad \!\mathbf{H}^\top\mathbf{S}\left(\mathbf{S}^\top(\mathbf{R} + \mathbf{C}\mathbf{Y}^{\bullet+}\mathbf{C}^\top)\mathbf{S}\right)^{-1}\mathbf{S}^\top\\
	\mathbf{B}^\bullet  &= 	\mathbf{K}^\bullet (\mathbf{R} + \mathbf{C}\mathbf{Y}^{\bullet+}\mathbf{C}^\top)\mathbf{K}^{\bullet \top} \\
	&= \left(\mathbf{H}^\top\mathbf{S}\left(\mathbf{S}^\top(\mathbf{R} + \mathbf{C}\mathbf{Y}^{\bullet+}\mathbf{C}^\top)\mathbf{S}\right)^{-1}\mathbf{S}^\top\mathbf{H}\right)^{-1}%
\end{align*}
is a solution to \eqref{eq:OCI_mod_prob}. Furthermore, $(\mathbf{K}^\circ, \mathbf{B}^\circ)$ is in the feasible set of the stricter problem \eqref{eq:OCI_mod_prob} since $\mathbf{K}^\circ \mathbf{H} = \mathbf{I}$ by hypothesis and, as shown before, $\mathbf{K}^\circ\mathbf{C}\mathbf{V}_\perp = \mathbf{0}$, and $\mathbf{B}^\circ \succeq \mathbf{K}^\circ(\mathbf{R}+\mathbf{C} \mathbf{Y}^{\bullet+}\mathbf{C}^\top)\mathbf{K}^{\circ\top}$. Therefore, $J(\mathbf{B}^\bullet )\leq J(\mathbf{B}^\circ)$. By Proposition~\ref{prop:matrix_id}, at the end of this section, $\mathbf{B}^\bullet$ can be rewritten as $	\mathbf{B}^\bullet =	(\mathbf{H}^\top\mathbf{R}^{-1}\mathbf{H}-\mathbf{U}^\bullet)^{-1}$, with $\mathbf{U}^\bullet  = \mathbf{H}^\top\mathbf{R}^{-1}\mathbf{C}(\mathbf{Y}^\bullet+\mathbf{C}^\top \mathbf{R}^{-1}\mathbf{C})^+ \mathbf{C}^\top\mathbf{R}^{-1}\mathbf{H}$. Now notice that the triple $(\mathbf{Y}^\bullet,\mathbf{U}^\bullet, \mathbf{B}^\bullet)$ is in the feasible domain of \eqref{eq:OCI_Y}. Specifically, from previous analysis $\mathbf{Y}^\bullet \preceq \mathbf{P}^{-1}$ for all $\mathbf{P}\in \Pcal$, so constraint \eqref{eq:OCI_Y_Y} is satisfied; $\mathbf{Y}^\bullet +\mathbf{C}^\top \mathbf{R}^{-1}\mathbf{C} \succeq \mathbf{0}$, $\col{\mathbf{C}^\top\mathbf{R}^{-1}\mathbf{H}}\subseteq \col{(\mathbf{Y}^\bullet+\mathbf{C}^\top \mathbf{R}^{-1}\mathbf{C})}$, and $\mathbf{U}^{\bullet}- \mathbf{H}^\top\mathbf{R}^{-1}\mathbf{C}(\mathbf{Y}^\bullet+\mathbf{C}^\top \mathbf{R}^{-1}\mathbf{C})^+ \mathbf{C}^\top\mathbf{R}^{-1}\mathbf{H} = \mathbf{0} \succeq  \mathbf{0}$, so, by Proposition~\ref{prop:schur}\eqref{schur:Apsd}, constraint \eqref{eq:OCI_Y_U} is satisfied; $\mathbf{B}^\bullet = (\mathbf{H}^\top\mathbf{R}^{-1}\mathbf{H}-\mathbf{U}^\bullet)^{-1} \succ \mathbf{0}$, so, by Proposition~\ref{prop:schur}\eqref{schur:BIl}, constraint \eqref{eq:OCI_Y_B} is satisfied, thereby establishing the second statement.%
\end{proof}%

First, we prove the first statement of the theorem. Let  $(\mathbf{Y}^\star,\mathbf{U}^\star,\mathbf{B}^\star)$ be a solution to \eqref{eq:OCI_Y}. Therefore, $(\mathbf{Y}^\star,\mathbf{U}^\star,\mathbf{B}^\star)$ is in the feasible domain of \eqref{eq:OCI_Y} and, by Lemma~\ref{lem:feas_lem}, $(\mathbf{K}^\star,\mathbf{B}^\star)$ is in the feasible domain of \eqref{eq:OCI_orig_prob}. We show that $(\mathbf{K}^\star,\mathbf{B}^\star)$ is a solution to \eqref{eq:OCI_orig_prob} by contradiction. Suppose, by contradiction, that there is $(\mathbf{K}^{\circ},\mathbf{B}^\circ)$ in the feasible domain of \eqref{eq:OCI_orig_prob} such that $J(\mathbf{B}^\circ) < J(\mathbf{B}^\star)$. By Lemma~\ref{lem:feas_lem}, it follows that the triple $(\mathbf{Y}^\bullet,\mathbf{U}^\bullet, \mathbf{B}^\bullet)$ is in the feasible domain of \eqref{eq:OCI_Y} and $J(\mathbf{B}^\bullet) \leq J(\mathbf{B}^\circ)$. Since $(\mathbf{Y}^\star,\mathbf{U}^\star,\mathbf{B}^\star)$ is a solution to \eqref{eq:OCI_Y}, then $J(\mathbf{B}^\star) \leq J(\mathbf{B}^\bullet) \leq J(\mathbf{B}^\circ)$. Bringing everything together yields $J(\mathbf{B}^\circ) < J(\mathbf{B}^\star)  \leq J(\mathbf{B}^\bullet) \leq J(\mathbf{B}^\circ)$, which is a contradiction.

Second, an analogous approach is used to prove the second statement of the theorem. Let $(\mathbf{K}^{\circ},\mathbf{B}^\circ)$ be a solution to \eqref{eq:OCI_orig_prob}. Therefore, $(\mathbf{K}^{\circ},\mathbf{B}^\circ)$ is in the feasible domain of  \eqref{eq:OCI_orig_prob} and, by Lemma~\ref{lem:feas_lem}, the triple $(\mathbf{Y}^\bullet,\mathbf{U}^\bullet, \mathbf{B}^\bullet)$ is in the feasible domain of \eqref{eq:OCI_Y} and $J(\mathbf{B}^\bullet) \leq J(\mathbf{B}^\circ)$. We show that  $(\mathbf{Y}^\bullet,\mathbf{U}^\bullet, \mathbf{B}^\bullet)$ is a solution to \eqref{eq:OCI_Y} by contradiction. Suppose, by contradiction, that there is  $(\mathbf{Y}^\star,\mathbf{U}^\star,\mathbf{B}^\star)$ in the feasible domain of \eqref{eq:OCI_Y} such that $J(\mathbf{B}^\star) < J(\mathbf{B}^\bullet)$. By Lemma~\ref{lem:feas_lem}, it follows that $(\mathbf{K}^\star,\mathbf{B}^\star)$ is in the feasible domain of \eqref{eq:OCI_orig_prob}. Since $(\mathbf{K}^{\circ},\mathbf{B}^\circ)$ is a solution to \eqref{eq:OCI_orig_prob}, then $J(\mathbf{B}^\circ)\leq J(\mathbf{B}^\star)$. Bringing everything together yields $J(\mathbf{B}^\star) < J(\mathbf{B}^\bullet) \leq J(\mathbf{B}^\circ) \leq J(\mathbf{B}^\star)$, which is a contradiction.

Finally, the last statement of the theorem, i.e., that the OCI problem \eqref{eq:OCI_orig_prob} is feasible if and only if  \eqref{eq:OCI_Y} is feasible, follows immediately from Lemma~\ref{lem:feas_lem}.

\begin{proposition}\label{prop:matrix_id} Let $\mathbf{R} \in S_{++}^o$, $\mathbf{C} \in \R^{o\times m}$, and $\mathbf{Y} \in S_+^m$. By the spectral theorem for real symmetric matrices,$ \mathbf{Y}$ admits a factorization
\begin{equation}\label{eq:factorization_Y}
	\mathbf{Y} = [\mathbf{V}\; \mathbf{V}_\perp]\begin{bmatrix}
		\mathbf{D} & \mathbf{0} \\ \mathbf{0}_{(m-r)\times r} & \mathbf{0}
	\end{bmatrix}[\mathbf{V}\; \mathbf{V}_\perp]^\top,
\end{equation}
where $r = \rank(\mathbf{Y})$. Define $\mathbf{S}_\perp$ as a matrix whose columns form an orthonormal basis for $\mathbf{C}\mathbf{V}_\perp$ and  $\mathbf{S}$ such that the columns of $[\mathbf{S}\; \mathbf{S}_\perp]$ form an orthonormal basis for $\R^o$. Then, $\mathbf{R}^{-1}-\mathbf{R}^{-1}\mathbf{C}(\mathbf{Y}+\mathbf{C}^\top \mathbf{R}^{-1}\mathbf{C})^+ \mathbf{C}^\top\mathbf{R}^{-1} = \mathbf{S}(\mathbf{S}^\top(\mathbf{R} + \mathbf{C}\mathbf{Y}^{+}\mathbf{C}^\top)\mathbf{S})^{-1}\mathbf{S}^\top$.
\end{proposition}
\begin{proof}
	Since $\mathbf{R}\succ \mathbf{0}$, a vector $\mathbf{v}\in \R^m$ is an eigenvector with null eigenvalue of $\mathbf{Y} + \mathbf{C}^\top\mathbf{R}^{-1}\mathbf{C}$ if and only if $\mathbf{v} \in \col \mathbf{V}_\perp$ and $\mathbf{v}\in \ker \mathbf{C}$. Furthermore, that is equivalent to $\mathbf{v} = \mathbf{V}_\perp \mathbf{x}$ for some $\mathbf{x} \in \ker \mathbf{C}\mathbf{V}_\perp$. One concludes that the eigenspace of $\mathbf{Y} + \mathbf{C}^\top\mathbf{R}^{-1}\mathbf{C}$ associated with the null eigenvalue is $E_0 = \{\mathbf{v}\in \R^m : \mathbf{v} = \mathbf{V}_\perp \mathbf{x} \;\text{for some}\; \mathbf{x} \in \ker \mathbf{C}\mathbf{V}_\perp\}$. Therefore,  by the spectral theorem for real symmetric matrices, $\mathbf{Y} + \mathbf{C}^\top\mathbf{R}^{-1}\mathbf{C}$ admits a factorization
	\begin{equation}
		\!\!\mathbf{Y} \!+ \mathbf{C}^\top\mathbf{R}^{-1}\mathbf{C} = [\tilde{\mathbf{V}}\; \tilde{\mathbf{V}}_\perp]\!\!\begin{bmatrix}
			\tilde{\mathbf{D}} \!\!& \mathbf{0} \\ \mathbf{0}\!\!& \mathbf{0}_{(m-\tilde{r})\times (m-\tilde{r})}
		\end{bmatrix}\!\![\tilde{\mathbf{V}}\; \tilde{\mathbf{V}}_\perp]^\top\!\!,\!\!
	\end{equation} 
	 where  $\tilde{r} = \rank(\mathbf{Y} + \mathbf{C}^\top\mathbf{R}^{-1}\mathbf{C})$, $\tilde{\mathbf{V}}_\perp \in \R^{m\times (m-\tilde{r})}$ is such that $\col \tilde{\mathbf{V}}_\perp = E_0$, $\tilde{\mathbf{V}} \in \R^{m\times \tilde{r}}$ is such that the columns of $[\tilde{\mathbf{V}}\; \tilde{\mathbf{V}}_\perp]$ form an orthonormal basis for $\R^m$, and $\tilde{\mathbf{D}} \in \R^{\tilde{r} \times \tilde{r}}$ is a diagonal matrix. Furthermore, by the definition of $\tilde{\mathbf{V}}_\perp$, it follows that $\mathbf{C}\tilde{\mathbf{V}}_\perp = \mathbf{0}$. Then, one can write for all $\epsilon >0$
	\begin{equation*}
		\begin{split}
			\mathbf{C}(\mathbf{Y} \!\!+\! \mathbf{C}^\top\!\mathbf{R}^{-1}\mathbf{C})^+\mathbf{C}^\top \!\!&= [\mathbf{C}\tilde{\mathbf{V}}\; \mathbf{0}]\begin{bmatrix}
				\tilde{\mathbf{D}}^{-1} & \mathbf{0}\\ \mathbf{0} & \frac{1}{\epsilon}\mathbf{I}
			\end{bmatrix}[\mathbf{C}\tilde{\mathbf{V}}\; \mathbf{0}]^\top\\
		&=  \!\mathbf{C}[\tilde{\mathbf{V}}\; \tilde{\mathbf{V}}_\perp]\!\!\begin{bmatrix}
			\tilde{\mathbf{D}}^{-1} & \mathbf{0}\\ \mathbf{0} &  \frac{1}{\epsilon}\mathbf{I}
		\end{bmatrix}\!\![\tilde{\mathbf{V}}\; \tilde{\mathbf{V}}_\perp]^\top\mathbf{C}^\top\\
	& = \!\mathbf{C}(\mathbf{Y}\! \!+ \!\mathbf{C}^\top\!\mathbf{R}^{-1}\mathbf{C} \!+\! \epsilon \tilde{\mathbf{V}}_\perp\! \tilde{\mathbf{V}}_\perp^\top )^{-1}\mathbf{C}^\top\!\!\!.
	\end{split}
	\end{equation*}
Therefore, since $\mathbf{C}(\mathbf{Y}+\mathbf{C}^\top \mathbf{R}^{-1}\mathbf{C})^+\mathbf{C}^\top = \mathbf{C}(\mathbf{Y} + \mathbf{C}^\top\mathbf{R}^{-1}\mathbf{C} +\epsilon \tilde{\mathbf{V}}_\perp \tilde{\mathbf{V}}_\perp^\top )^{-1}\mathbf{C}^\top$ holds for any $\epsilon >0$ and $\col\tilde{\mathbf{V}}_\perp = E_0 \subseteq \col \mathbf{\mathbf{V}_\perp}$, one can conclude that $\lim _{\epsilon \to 0}  \mathbf{C}(\mathbf{Y} + \mathbf{C}^\top\mathbf{R}^{-1}\mathbf{C} + \epsilon \mathbf{V}_\perp \mathbf{V}_\perp^\top )^{-1}\mathbf{C}^\top$ exists and is equal to $\mathbf{C}(\mathbf{Y} + \mathbf{C}^\top\mathbf{R}^{-1}\mathbf{C})^+\mathbf{C}^\top$. 
From the Sherman-Morrison-Woodbury formula \cite[Section~2.1.4]{GolubVanLoan2013} $ \mathbf{R}^{-1}-\mathbf{R}^{-1}\mathbf{C}(\mathbf{Y}+\mathbf{C}^\top \mathbf{R}^{-1}\mathbf{C} + \epsilon \mathbf{V}_\perp \mathbf{V}_\perp^\top)^{-1} \mathbf{C}^\top\mathbf{R}^{-1} = (\mathbf{R}+\mathbf{C}(\mathbf{Y}+\epsilon \mathbf{V}_\perp \mathbf{V}_\perp^\top)^{-1}\mathbf{C^\top})^{-1}$. Therefore,
\begin{equation}\label{eq:matrix_id_aux_1}
 	\begin{split}
 		\lim_{\epsilon \to 0} (\mathbf{R}+\mathbf{C}&(\mathbf{Y}+\epsilon \mathbf{V}_\perp \mathbf{V}_\perp^\top)^{-1}\mathbf{C^\top})^{-1} \\ = &	\mathbf{R}^{-1}-\mathbf{R}^{-1}\mathbf{C}(\mathbf{Y}+\mathbf{C}^\top \mathbf{R}^{-1}\mathbf{C})^+ \mathbf{C}^\top\mathbf{R}^{-1}.
 	\end{split}
 \end{equation}%
 Since $\mathbf{Y}$ admits the factorization \eqref{eq:factorization_Y} and, by the definition of $\mathbf{S}_\perp$, $\mathbf{S}^\top\mathbf{C}\mathbf{V}_\perp = \mathbf{0}$, it follows that
\begin{equation*}
	\begin{split}
		&(\mathbf{R}+\mathbf{C}(\mathbf{Y}+\epsilon \mathbf{V}_\perp \mathbf{V}_\perp^\top)^{-1}\mathbf{C}^\top)^{-1} \\
		&\!\!=\! 
			[\mathbf{S}\;\mathbf{S}_\perp]\!\!\left(\!\begin{bmatrix}\mathbf{S}^\top \\ \mathbf{S}_\perp^\top\end{bmatrix}\!(\mathbf{R}\!+\!\mathbf{C}(\mathbf{Y}\!+\!\epsilon \mathbf{V}_\perp \mathbf{V}_\perp^\top)^{-1}\mathbf{C}^\top)[\mathbf{S}\;\mathbf{S}_\perp]\!\!\right)^{\!\!-1}\!\!\begin{bmatrix}\mathbf{S}^\top \\ \mathbf{S}_\perp^\top\end{bmatrix} \\
			&\!\!= \![\mathbf{S}\;\mathbf{S}_\perp]\begin{bmatrix} \mathbf{M}_{11}
		 & 	\mathbf{M}_{12}\\
			\mathbf{M}_{12}^\top & \mathbf{M}_{22}
			\end{bmatrix}^{-1} [\mathbf{S}\;\mathbf{S}_\perp]^\top,
	\end{split} 	
\end{equation*}
where
\begin{equation*}
	\begin{split}
		\mathbf{M}_{11} &= 	\mathbf{S}^\top(\mathbf{R}+\mathbf{C}\mathbf{V}\mathbf{D}^{-1}(\mathbf{C}\mathbf{V})^\top)\mathbf{S}\\
		\mathbf{M}_{12}&=  \mathbf{S}^\top(\mathbf{R}+\mathbf{C}\mathbf{V}\mathbf{D}^{-1}(\mathbf{C}\mathbf{V})^\top)\mathbf{S}_{\perp}\\
		\mathbf{M}_{22} & = 	\mathbf{S}_{\perp}^\top(\mathbf{R}+\mathbf{C}\mathbf{V}\mathbf{D}^{-1}(\mathbf{C}\mathbf{V})^\top + \frac{1}{\epsilon}\mathbf{C}\mathbf{V}_{\perp}(\mathbf{C}\mathbf{V}_\perp)^\top )\mathbf{S}_{\perp}.
	\end{split}
\end{equation*}
As a result, the limit as $\epsilon \to 0$ of $(\mathbf{R}+\mathbf{C}(\mathbf{Y}+\epsilon \mathbf{V}_\perp \mathbf{V}_\perp^\top)^{-1}\mathbf{C}^\top)^{-1}$ exists and is equal to 
\begin{equation}\label{eq:matrix_id_aux_2}
	\begin{split}
		 \lim_{\epsilon \to 0} (\mathbf{R}+\mathbf{C}(\mathbf{Y}+\epsilon &\mathbf{V}_\perp \mathbf{V}_\perp^\top)^{-1}\mathbf{C^\top})^{-1} \\
		 & = 	[\mathbf{S}\;\mathbf{S}_\perp]\!\begin{bmatrix}
		 	\mathbf{M}_{11}^{-1} & \mathbf{0}\\
		 	\mathbf{0}&\mathbf{0}
		 \end{bmatrix}\![\mathbf{S}\;\mathbf{S}_\perp]^\top\\
	 &  = \mathbf{S}(\mathbf{S}^\top(\mathbf{R}+\mathbf{C}\mathbf{Y}^+\mathbf{C}^\top)\mathbf{S})^{-1}\mathbf{S}^\top.
	\end{split}
\end{equation}%
Comparing \eqref{eq:matrix_id_aux_1} and \eqref{eq:matrix_id_aux_2} concludes the proof.
\end{proof}
\subsection{Proof of Theorem~\ref{th:feas_iff}}\label{sec:proof_feas_iff}
A necessary condition for each of the two statements between which an equivalence is established in this result is that  $\mathbf{H}$ is full column rank. Therefore, for the remainder of the proof, $\mathbf{H}$ is assumed to be full column rank. The proof of the theorem relies on the following proposition.

\begin{proposition}\label{prop:help_feas_iff}
	For a given $\mathbf{Y}$, the LMI defined by \eqref{eq:OCI_Y_B}-\eqref{eq:OCI_Y_U} is feasible if and only if $\mathbf{H}^\top\mathbf{R}^{-1}\mathbf{H}-\mathbf{H}^\top\mathbf{R}^{-1}\mathbf{C}(\mathbf{Y}+\mathbf{C}^\top \mathbf{R}^{-1}\mathbf{C})^+ \mathbf{C}^\top\mathbf{R}^{-1}\mathbf{H}$ is full rank. Furthermore, if for a given $\mathbf{Y}$ the expression above is full rank, then it is also full rank if one replaces $\mathbf{Y}$ by any $\mathbf{Y}^\prime$ that satisfies $\col \mathbf{Y} \subseteq \col{\mathbf{Y}^\prime}$.
\end{proposition}
\begin{proof}
	We start with the first statement. In one direction, if the LMI defined by \eqref{eq:OCI_Y_B}-\eqref{eq:OCI_Y_U} is feasible then there exist $\mathbf{U}$ and $\mathbf{B}$ such that, by Proposition~\ref{prop:schur}\eqref{schur:BIr}, $\mathbf{H}^\top\mathbf{R}^{-1}\mathbf{H}-\mathbf{U} \succ \mathbf{0}$ and, by Proposition~\ref{prop:schur}\eqref{schur:Apsd}, $\mathbf{U}\succeq 	\mathbf{H}^\top\mathbf{R}^{-1}\mathbf{C}(\mathbf{Y}+\mathbf{C}^\top \mathbf{R}^{-1}\mathbf{C})^+ \mathbf{C}^\top\mathbf{R}^{-1}\mathbf{H}$, thus $\mathbf{H}^\top\mathbf{R}^{-1}\mathbf{H}-\mathbf{H}^\top\mathbf{R}^{-1}\mathbf{C}(\mathbf{Y}+\mathbf{C}^\top \mathbf{R}^{-1}\mathbf{C})^+ \mathbf{C}^\top\mathbf{R}^{-1}\mathbf{H} \succ \mathbf{0}$.  In the other direction, if $\mathbf{H}^\top\mathbf{R}^{-1}\mathbf{H}-\mathbf{H}^\top\mathbf{R}^{-1}\mathbf{C}(\mathbf{Y}+\mathbf{C}^\top \mathbf{R}^{-1}\mathbf{C})^+ \mathbf{C}^\top\mathbf{R}^{-1}\mathbf{H} \succ \mathbf{0}$, one can choose $\mathbf{U} = \mathbf{H}^\top\mathbf{R}^{-1}\mathbf{C}(\mathbf{Y}+\mathbf{C}^\top \mathbf{R}^{-1}\mathbf{C})^+ \mathbf{C}^\top\mathbf{R}^{-1}\mathbf{H}$, thus  $\mathbf{H}^\top\mathbf{R}^{-1}\mathbf{H}-\mathbf{U} \succ \mathbf{0}$ and choosing $\mathbf{B} = (\mathbf{H}^\top\mathbf{R}^{-1}\mathbf{H}-\mathbf{U} )^{-1}$, \eqref{eq:OCI_Y_B} is satisfied by Proposition~\ref{prop:schur}(\ref{schur:BIl}). Moreover, since $\mathbf{Y}+\mathbf{C}^\top \mathbf{R}^{-1}\mathbf{C} \succeq \mathbf{0}$, $\col{\mathbf{C}^\top\mathbf{R}^{-1}\mathbf{H}}\subseteq \col{(\mathbf{Y}+\mathbf{C}^\top \mathbf{R}^{-1}\mathbf{C}})$, and $\mathbf{U}- \mathbf{H}^\top\mathbf{R}^{-1}\mathbf{C}(\mathbf{Y}+\mathbf{C}^\top \mathbf{R}^{-1}\mathbf{C})^+ \mathbf{C}^\top\mathbf{R}^{-1}\mathbf{H} = \mathbf{0} \succeq  \mathbf{0}$, by Proposition~\ref{prop:schur}\eqref{schur:Apsd} it follows that \eqref{eq:OCI_Y_U} is satisfied. To prove the second statement, recall, from Proposition~\ref{prop:matrix_id}, which is in the proof of Theorem~\ref{th:equiv_OCI_prob} in Appendix~\ref{sec:proof_lem_equiv_OCI_prob}, that
	\begin{equation}\label{eq:mp2S_feas}
		\begin{split}
		&\mathbf{H}^\top\mathbf{R}^{-1}\mathbf{H}-\mathbf{H}^\top\mathbf{R}^{-1}\mathbf{C}(\mathbf{Y}\!+\!\mathbf{C}^\top \mathbf{R}^{-1}\mathbf{C})^+ \mathbf{C}^\top\mathbf{R}^{-1}\mathbf{H}\!\!\\
		\!\!=\;& \mathbf{H}^\top\mathbf{S}(\mathbf{S}^\top(\mathbf{R} + \mathbf{C}\mathbf{Y}^{+}\mathbf{C}^\top)\mathbf{S})^{-1}\mathbf{S}^\top\mathbf{H},
	\end{split}
	\end{equation}
	where $\mathbf{S}$ is defined such that its columns form an orthonormal basis for $(\col (\mathbf{C}\mathbf{V}_\perp))^\perp$ and $\mathbf{V}_\perp$ is defined such that its columns form an orthonormal basis for $(\col{\mathbf{Y}})^\perp$. Therefore, given a $\mathbf{Y}$ for which \eqref{eq:mp2S_feas} is full rank, any $\mathbf{Y}^\prime$ such that $\col \mathbf{Y} \subseteq \col{\mathbf{Y}^\prime}$ implies $\col \mathbf{S} \subseteq \col \mathbf{S}^\prime $, so replacing $\mathbf{Y}$ with $\mathbf{Y}^\prime$ in \eqref{eq:mp2S_feas} does not change its rank.
\end{proof}

An equivalence between the feasibility of the OCI problem \eqref{eq:OCI_orig_prob} and \eqref{eq:OCI_Y} was already established in Lemma~\ref{lem:feas_lem}. Therefore, to establish the theorem, one only needs to show the equivalence between Condition~\ref{cond:feas}, i.e., $\mathbf{H}^\top\mathbf{R}^{-1}\mathbf{H}-\mathbf{H}^\top\mathbf{R}^{-1}\mathbf{C}(\mathbf{W}^\top\mathbf{W}+\mathbf{C}^\top \mathbf{R}^{-1}\mathbf{C})^+ \mathbf{C}^\top\mathbf{R}^{-1}\mathbf{H}$ being full rank, and the feasibility of \eqref{eq:OCI_Y} using Proposition~\ref{prop:help_feas_iff}.
Recall, from Proposition~\ref{prop:big_bound_on_P} in the proof of Lemma~\ref{lem:boundedness}, that $\mathbf{P}^{-1} \succeq \mathbf{W}^\top\diag(\mathbf{X}_1, \mathbf{X}_2, \ldots,\mathbf{X}_M)^{-1}\mathbf{W}/M$ for all $\mathbf{P}\in \Pcal$.
On the one hand, if Condition~\ref{cond:feas} holds, then $\mathbf{Y} = \mathbf{W}^\top\diag(\mathbf{X}_1, \mathbf{X}_2, \ldots,\mathbf{X}_M)^{-1}\mathbf{W}/M$ has the same column space as $\mathbf{W}^\top$. Therefore, by Proposition~\ref{prop:help_feas_iff}, $\mathbf{H}^\top\mathbf{R}^{-1}\mathbf{H}-\mathbf{H}^\top\mathbf{R}^{-1}\mathbf{C}(\mathbf{Y}+\mathbf{C}^\top \mathbf{R}^{-1}\mathbf{C})^+ \mathbf{C}^\top\mathbf{R}^{-1}\mathbf{H}$ is full rank. Additionally, $\mathbf{Y} = \mathbf{W}^\top\diag(\mathbf{X}_1, \mathbf{X}_2, \ldots,\mathbf{X}_M)^{-1}\mathbf{W}/M \preceq \mathbf{P}^{-1}$ for all $\mathbf{P}\in \Pcal$ so, by Proposition~\ref{prop:help_feas_iff}, \eqref{eq:OCI_Y} is feasible. 
On the other hand, if \eqref{eq:OCI_Y} is feasible, then there exists $\mathbf{Y}$ such that $\mathbf{Y} \preceq \mathbf{P}^{-1}$ for all $\mathbf{P}\in \Pcal$ and, by Proposition~\ref{prop:help_feas_iff}, $\mathbf{H}^\top\mathbf{R}^{-1}\mathbf{H}-\mathbf{H}^\top\mathbf{R}^{-1}\mathbf{C}(\mathbf{Y}+\mathbf{C}^\top \mathbf{R}^{-1}\mathbf{C})^+ \mathbf{C}^\top\mathbf{R}^{-1}\mathbf{H}$ is full rank. Since the bounded components of $\mathbf{P}$ are characterized by the row space of $\mathbf{W}$, it follows that $\col{\mathbf{Y}}\subseteq \col{\mathbf{W}^\top\mathbf{W}}$, therefore, by Proposition~\ref{prop:help_feas_iff}, $\mathbf{H}^\top\mathbf{R}^{-1}\mathbf{H}-\mathbf{H}^\top\mathbf{R}^{-1}\mathbf{C}(\mathbf{W}^\top\mathbf{W}+\mathbf{C}^\top \mathbf{R}^{-1}\mathbf{C})^+ \mathbf{C}^\top\mathbf{R}^{-1}\mathbf{H}$ is full rank, i.e., Condition~\ref{cond:feas} holds.

\subsection{Proof of Corollary~\ref{cor:OCI_Kahan_optimal}}\label{sec:proof_cor_OCI_Kahan_optimal}

Notice that it is possible to decouple the optimization of $\boldsymbol{\omega}$ from the remaining decision variables in \eqref{eq:OCI_kahan}. Specifically, if one fixes parameter $\boldsymbol{\omega}$, one can approach \eqref{eq:OCI_kahan} resorting to Theorem~\ref{th:equiv_OCI_prob}. Then, choosing a parameter $\boldsymbol{\omega}$ that yields the minimum objective at the optimizer decouples the problem. Therefore, the first statement of the corollary follows immediately from Theorem~\ref{th:equiv_OCI_prob}. To prove the second statement, on the one hand, notice that if a triple $(\mathbf{U},\mathbf{B},\boldsymbol{\omega})$ is feasible for \eqref{eq:OCI_Kahan_optimal}, then the triple  $(\mathbf{U},\mathbf{B},\mathbf{Y})$ with  $\mathbf{Y} = \sum_{b=1}^M\boldsymbol{\omega}_b\mathbf{Y}_b$ is immediately feasible for \eqref{eq:OCI_Y}, and, by Theorems~\ref{th:equiv_OCI_prob} and~\ref{th:feas_iff}, it follows that Condition~\ref{cond:feas} holds. On the other hand, let Condition~\ref{cond:feas} hold, i.e., $\mathbf{H}^\top\mathbf{R}^{-1}\mathbf{H}-\mathbf{H}^\top\mathbf{R}^{-1}\mathbf{C}(\mathbf{W}^\top\mathbf{W}+\mathbf{C}^\top \mathbf{R}^{-1}\mathbf{C})^+ \mathbf{C}^\top\mathbf{R}^{-1}\mathbf{H}$ is full rank. Then for $\boldsymbol{\omega}_b = 1/M$ for all $b = 1,2,\ldots, M$, $\mathbf{Y} = \sum_{b=1}^M\boldsymbol{\omega}_b\mathbf{Y}_b$ has the same column space as $\mathbf{W}^\top\mathbf{W}$. Therefore, by Proposition~\ref{prop:help_feas_iff}, $\mathbf{H}^\top\mathbf{R}^{-1}\mathbf{H}-\mathbf{H}^\top\mathbf{R}^{-1}\mathbf{C}(\mathbf{Y}+\mathbf{C}^\top \mathbf{R}^{-1}\mathbf{C})^+ \mathbf{C}^\top\mathbf{R}^{-1}\mathbf{H}$ is full rank and there is a pair $(\mathbf{U},\mathbf{B})$ that satisfies the LMIs \eqref{eq:OCI_Y_B}-\eqref{eq:OCI_Y_U} for this choice of $\mathbf{Y}$. It follows immediately that $(\mathbf{U},\mathbf{B},\boldsymbol{\omega})$  is feasible for \eqref{eq:OCI_Kahan_optimal}.

\section*{Acknowledgments}
We thank an anonymous reviewer for identifying a minor technical mistake in an early version of the proof of Lemma~\ref{lem:boundedness}.

\section*{References}
\bibliographystyle{IEEEtran}
\bibliography{references-c,bibliography}

\end{document}